\documentclass{amsart}

\usepackage[T1]{fontenc}
\usepackage[USenglish]{babel}
\usepackage{graphicx}
\usepackage{natbib}
\usepackage{subfig}
\usepackage{bm}
\usepackage{dsfont}
\usepackage{subfig}
\usepackage{fontawesome}
\usepackage[colorlinks,citecolor=blue,urlcolor=blue]{hyperref}
\newtheorem{lemma}{Lemma}
\newtheorem{corollary}{Corollary}
\newcommand{\diag}{\mathop{\mathrm{diag}}}

\begin{document}

\title[Steady-state Behavior of MEWMA]{The Steady-State Behavior of Multivariate Exponentially Weighted Moving Average Control Charts}
	
\author[S. Knoth]{Sven Knoth}
			
\email{knoth@hsu-hh.de}
			
\address{Department of Mathematics and Statistics, Helmut Schmidt University,\\ University of the Federal Armed Forces Hamburg,\\ Holstenhofweg 85, 22043 Hamburg, Germany}

\begin{abstract}
\ Multivariate Exponentially Weighted Moving Average, MEWMA, charts are
popular, handy and effective procedures to detect distributional changes in a stream of multivariate data.
For doing appropriate performance analysis, dealing with the steady-state behavior of the MEWMA statistic is essential.
Going beyond early papers, we derive quite accurate approximations of
the respective steady-state densities of the MEWMA statistic.
It turns out that these densities could be rewritten as the product of two functions depending on one argument only
which allows feasible calculation.
For proving the related statements, the presentation of the
non-central chisquare density deploying the 
confluent hypergeometric limit function is applied.
Using the new methods it was found that
for large dimensions, the steady-state behavior becomes different to what one might expect from the univariate monitoring field.
Based on the integral equation driven methods,
steady-state and worst-case average run lengths are calculated with higher accuracy than before.
Eventually, optimal MEWMA smoothing constants are derived for all considered measures.
\end{abstract}

\keywords{Multivariate Statistical Process Control; Fredholm Integral Equation of the Se\-cond Kind; Nystr\"om Method; Markov Chain Approximation; Non-Central Chisquare Distribution}

\maketitle


\section{Introduction}

Multivariate monitoring tasks result often in some type of Multivariate Exponentially Weighted Moving Average (MEWMA) which was
introduced by \cite{Lowr:Wood:Cham:Rigd:1992} as extension of the even more popular $T^2$ chart proposed initially by \cite{Hote:1947}.
Refer to \cite{Yang:Huan:Lai:Jin:2018} and \cite{Harr:EtAl:2018} for recent applications of MEWMA
in the field of  fault detection in wind turbines and photovoltaic systems, respectively.
In a nutshell, MEWMA charts aim to detecting changes in the distribution (here in the mean) of multivariate data as quickly
as possible while maintaining a reasonable level of false alarms.
The most common operating characteristic of a 
monitoring device alias control chart is the Average Run Length (ARL) introduced already in \cite{Page:1954a}. 
Its typical appearance is often called \textit{zero-state} ARL and refers to the situation
that the state of the control chart at the time of change is known.
To describe this more thoroughly, we take a look at our data model.
Here we consider a sequence of serially independent normally distributed vectors $\bm{X}_1, \bm{X}_2, \ldots$ of dimension $p$, that is
\begin{equation*}
  \bm{X}_n\sim \mathcal{N}(\bm{\mu}, \Sigma) \;\;,\; n = 1,2,\ldots
\end{equation*}
To avoid further complications, we assume that the covariance matrix $\Sigma$ is known,
and the mean vector $\bm{\mu}$ follows the simple change point model: $\bm{\mu} = \bm{\mu}_0$ for $n < \tau$, and $\bm{\mu} = \bm{\mu}_1$ for $n \ge \tau$.
The change point $\tau$ is, of course, unknown, while $\bm{\mu}_0$ is given (either by knowing the process or by estimating during a preliminary study).
The other mean value, $\bm{\mu}_1$,
induces certain choices of control chart parameters. 
Following \cite{Lowr:Wood:Cham:Rigd:1992}, the MEWMA sequence $\{\bm{Z}_n\}$ is formed by
\begin{equation}
  \bm{Z}_0 = \bm{z}_0 \;,\;\; \bm{Z}_n = (1-\lambda) \bm{Z}_{n-1} + \lambda \bm{X}_n \quad,\; n = 1, 2, \ldots \;,\;\; 0<\lambda\le 1 \,. \label{eq:mewma1}
\end{equation}
In parallel, we determine the Mahalanobis distance $T_n^2 = (\bm{Z}_n-\bm{\mu}_0)^\prime \Sigma_Z^{-1} (\bm{Z}_n-\bm{\mu}_0)$
from the stable mean $\bm{\mu}_0$,
where $\Sigma_Z$ denotes the asymptotic covariance matrix of $\bm{Z}_n$ with
\begin{equation*}
  \Sigma_Z = \lim_{n\to\infty} Cov(\bm{Z}_n) = \frac{\lambda}{2-\lambda} \Sigma \,.
\end{equation*}
If this distance, $T_n^2$, becomes larger than a given threshold $h_4$ \citep[naming convention stems from][]{Lowr:Wood:Cham:Rigd:1992}, an alarm is triggered which is linked to
the MEWMA stopping time
\begin{equation}
  N = \inf \big\{ n\ge 1: T_n^2 > h_4 \big\} \,. \label{eq:N}
\end{equation}
Its expected value for two exemplary cases, $\tau=1$ or $\tau=\infty$, is just the aforementioned
\textit{zero-state} average run length (ARL), roughly speaking.
In the sequel, this is written as $E_\infty(N)$ (\textit{in-control} case) and $E_1(N)$ (\textit{out-of-control} case)
with the general expression $E_\tau()$ denoting the expectation for given change point $\tau$.
In order to obtain actual numbers, \cite{Lowr:Wood:Cham:Rigd:1992} deployed Monte Carlo simulations,
\cite{Rigd:1995a, Rigd:1995b} provided numerical solutions of ARL integral equations, and
\cite{Rung:Prab:1996} presented a Markov chain approximation.
Recently, \cite{Knot:2017a} demonstrated some accuracy problems of these algorithms and
offered improved numerical solutions of \cite{Rigd:1995a, Rigd:1995b}.
However, only the Monte Carlo and the Markov chain approach are expanded to determine the \textit{steady-state} ARL,
which measures the average number of observations until signal after the change point $\tau$, while
assuming that the sequence $\bm{Z}_n$ reached its steady state before $\tau$.
Namely, \cite{Prab:Rung:1997} utilized the Markov chain model to calculate the \textit{steady-state} ARL.
Their algorithm was used, for example, in \cite{Lee:Khoo:2014a}. However, its deployment is complicated
and, differently to the \textit{zero-state} ARL, no software implementation is published.
Hence, others used Monte Carlo studies, see, for example, \cite{Reyn:Stou:2008} and \cite{Zou:Tsun:2011}.
Before we start to investigate the \textit{steady-state} ARL in more detail, we want to emphasize
its importance as performance indicator of
a monitoring device. Because the actual position of the change point $\tau$ is unknown,
we do not know neither the position of the MEWMA statistic $\bm{Z}_{\tau-1}$ nor its distance to $\bm{\mu}_0$, $T_{\tau-1}^2$, one observation
before the change occurs.
For the mentioned \textit{zero-state} ARL we imply that $\bm{Z}_{\tau-1} = \bm{\mu}_0$ and $T_{\tau-1}^2 = 0$, respectively,
what might be substantially misleading. More appropriate would be to exploit the \textit{steady-state} behavior of
$\bm{Z}_{\tau-1}$ in order to weight in a reasonable way possible positions of $\bm{Z}_{\tau-1}$ and
the resulting detection delay $\sim N-\tau$.
More conservative would be to investigate the \textit{worst-case} position of $\bm{Z}_{\tau-1}$.
Both ways will be treated and finally compared to the classic Hotelling-Shewhart chart which remained popular for
monitoring users being afraid of inertia problems which often escort the application of (M)EWMA.
Fortunately, \cite{Rigd:1995a} indicated that it suffices to study the simple case $\bm{\mu}_0 = \bm{0}$ and $\Sigma = \mathbb{I}$
(identity matrix) by only assuming that the original covariance matrix $\Sigma$ is positive definite.
Hence, in the sequel we set both terms accordingly.

The paper is organized as follows: In Section \ref{sec:stst} the concepts of \textit{steady-state} ARL
are described in more detail while evaluating the \textit{in-control} case. The more involved and much more
important \textit{out-of-control} case is examined in Section \ref{sec:num}. In the subsequent Section \ref{sec:comp}
the framework is applied to illustrate the detection performance of MEWMA using the
\textit{zero-state}, \textit{steady-state} and \textit{worst-case} ARL.
Eventually, the conclusions section completes the paper. Proofs and similar technical
details are collected in the Appendix.

\section{Steady-state methodology and the in-control case}\label{sec:stst}

Measuring the detection delay after reaching some \textit{steady state} was already utilized in \cite{Robe:1966}.
Beginning with \cite{Tayl:1968} and later on with \cite{Yash:1985b} and \cite{Cros:1986}, the concepts were consolidated.
Using the naming conventions of \cite{Cros:1986}, two different types of \textit{steady-state} ARL
are defined in the following way.
The first and presumably more popular one assumes that no (false) alarm is raised before the change takes place.
It is called \textit{conditional steady-state} ARL and could be written as
\begin{equation}
  \mathcal{D} = \lim\limits_{\tau\to\infty} E_\tau \big(N-\tau+1\mid N\ge \tau\big) \,. \label{eq:D1}
\end{equation}
The second one refers to the situation that the change happens after a sequence of false alarms.
The control chart is re-started after each of them so that the \textit{cyclical steady-state} ARL
could be expressed by
\begin{align}
  \mathcal{D}_\star & = \lim\limits_{\tau\to\infty} E_\tau \big(N_\star-\tau+1) \label{eq:D2} \,, \\
  N_\star & =  N_1 + N_2 + \ldots + N_{I_\tau-1} + N_{I_\tau}
  \; , \,
  I_\tau = \min\left\{i\ge 1: \sum_{j=1}^i N_j \ge \tau\right\} . \nonumber
\end{align}
See \cite{Poll:Tart:2009b}, Section 3, for a more rigor treatment of $\mathcal{D}_\star$ and
asymptotic optimality in the univariate case.
Both \textit{steady-state} ARL types are calculated by combining the (quasi-)stationary distribution
of the control chart statistic and the ARL as function of the actual value of the latter
statistic. To develop our approach, we start with the simpler \textit{in-control} case, where it
is sufficient to consider for both functions only one argument, the distance $T_n^2$.
Recall the ARL integral equation of \cite{Rigd:1995a} with $\alpha = \bm{z}_0^\prime \bm{z_0}$ being
the distance of the initial $\bm{Z}_0$ value to zero in \eqref{eq:mewma1}:
\begin{equation}
  \mathcal{\mathring{L}}(\alpha)
  = 1 + \int_0^h \mathcal{\mathring{L}}(u) \frac{1}{\lambda^2} f_{\chi^2}\left(\frac{u}{\lambda^2} \,\Big|\, p,
  \left[\frac{1-\lambda}{\lambda}\right]^2\!\!\alpha\right) \mathrm{d}u \,,
  \label{eq:L0igl}
\end{equation}
where $\mathcal{\mathring{L}}(\alpha) = E_\infty(N)$ for $\bm{Z}_0 = \bm{z}_0$
(the superscript $\mathring{\,}$ marks the \textit{in-control} case)
and $h = h_4 \lambda/(2-\lambda)$.
The function $f_{\chi^2}(\cdot\mid p, \nu)$ denotes the probability density of the non-central $\chi^2$ distribution with $p$ degrees of freedom and noncentrality
parameter $\nu = \eta \alpha = \left[\frac{1-\lambda}{\lambda}\right]^2\!\!\alpha$.
\cite{Rigd:1995a} solved \eqref{eq:L0igl} numerically by applying the Nystr\"om method \citep{Nyst:1930}
with Gau\ss{}-Radau quadrature.
Recently, \cite{Knot:2017a} utilized the slightly more powerful Gau\ss{}-Legendre quadrature after
a change in variables from $\alpha$ to $\alpha^2$ which improves the accuracy for odd $p$ substantially.
A similar integral equation is valid for the left eigenfunction $\mathring{\psi}()$,
the quasi-stationary density of $Z_{\tau-1}$, which is needed for the 
\textit{conditional steady-state} ARL $\mathcal{D}$:
\begin{equation}
  \varrho \mathring{\psi}(u) =
  \int_0^h \mathring{\psi}(\alpha) \frac{1}{\lambda^2} f_{\chi^2}\left(\frac{u}{\lambda^2} \,\Big|\, p,\eta \alpha\right) \mathrm{d}\alpha \,.
  \label{eq:psi0aigl}
\end{equation}
Refer to \cite{Knot:2016a} for more details about the family of integral equations to calculate the ARL function
$\mathcal{\mathring{L}}()$ and the left eigenfunction $\mathring{\psi}()$ in case of univariate EWMA charts.
A similar list is given in \cite{Mous:Polu:Tart:2009} for CUSUM and Shiryaev-Roberts schemes.
The parameter $\varrho$ is just the dominating eigenvalue of the integral kernel in \eqref{eq:psi0aigl} which provides essential information about the long running behavior
of the MEWMA stopping time $N$ in the \textit{in-control} case --- $P_\infty(N = n\mid N\ge n) \approx 1 - \varrho$ for large $n$
\citep[classical paper is, for example,][]{Gold:1989}.
Applying the same change in variables as
performed in \cite{Knot:2017a} for \eqref{eq:L0igl}, one obtains
\begin{equation}
  \varrho \mathring{\psi}_i
  = \sum_{j=1}^r w_j \mathring{\psi}_j \frac{1}{\lambda^2} f_{\chi^2}\left(\frac{z_i^2}{\lambda^2} \,\Big|\, p, \eta z_j^2\right) 2 z_j
  \qquad\text{ with } \mathring{\psi}_i = \mathring{\psi}(z_i^2) \,,
  \label{eq:psi0ale}
\end{equation}
where $w_j$ and $z_j$ are the weights and nodes of the Gau\ss{}-Legendre quadrature. The system \eqref{eq:psi0ale}
could be solved either by the power method \citep{Mise:Poll:1929} or by applying readily available routines such
as \verb|eigen()| in the statistics software system \textsf{R} which calls eventually well established procedures from the
BLAS \citep{Laws:Hans:Kinc:Krog:1979} or LAPACK \citep{Ande:EtAl:1999} libraries (more information on \url{http://www.netlib.org}).
Collecting the numerical solutions of \eqref{eq:L0igl} and \eqref{eq:psi0ale} in matrices and vectors we can write
\begin{equation*}
  \bm{\mathring{\ell}} = (\mathbb{I} - \mathbb{Q}_\mathcal{L})^{-1} \bm{1}
  \quad \text{ and } \quad
  \varrho \bm{\mathring{\psi}} = \mathbb{Q}_\psi \bm{\mathring{\psi}} \,, \\
\end{equation*}
and obtain
\begin{equation*}
  \mathring{\mathcal{D}} = \frac{(\mathbb{W}\bm{\mathring{\psi}})^\prime \bm{\mathring{\ell}}}{(\mathbb{W}\bm{\mathring{\psi}})^\prime \bm{1}} \label{eq:D1approx}
  \quad \text{ with } \quad \mathbb{W} = \diag(2 w_i z_i)
\end{equation*}
as numerical counterpart to \eqref{eq:D1}.
More details are given in the Appendix.
Next, we derive the integral equation for the 
\textit{cyclical steady-state} ARL according to \eqref{eq:D2}. The now stationary density (linked to the eigenvalue $\varrho = 1$) follows
\begin{equation}
  \mathring{\psi}_\star(u) =
  \Psi_0 \frac{1}{\lambda^2} f_{\chi^2}\left(\frac{u}{\lambda^2} \,\Big|\, p\right) +
  \int_0^h \mathring{\psi}_\star(\alpha) \frac{1}{\lambda^2} f_{\chi^2}\left(\frac{u}{\lambda^2} \,\Big|\, p, \eta \alpha\right) \mathrm{d}\alpha \,,
  \label{eq:psi0bigl}
\end{equation}
which results after plugging in the Gau\ss{}-Legendre quadrature in a common linear equation system,
\begin{equation}
  \bm{\mathring{\psi}}_\star = (\mathbb{I} - \mathbb{Q}_\psi)^{-1} \bm{f}  \quad \text{ and } \quad \bm{f} = (f_1, \ldots, f_r)^\prime \;,\;\;
  f_i = \frac{\Psi_0}{\lambda^2} f_{\chi^2}\left(\frac{z_i^2}{\lambda^2} \,\Big|\, p\right) \,. \label{eq:psi0ble}
\end{equation}
The additional parameter $\Psi_0$ labels the probability of $\bm{Z}_{\tau-1} = \bm{z}_0 = \bm{0}$  for $\tau \to \infty$ and equals to $1/E_\infty(N)$
because the restart in $\bm{z}_0$ is a renewal point, cf. \cite{Knot:2016a} for more details. The solution in \eqref{eq:psi0ble} complies
$(\mathbb{W}\bm{\mathring{\psi}}_\star)^\prime \bm{1} + \Psi_0 = 1$ by construction so that we can write
\begin{equation*}
  \mathring{\mathcal{D}}_\star = \Psi_0 \, \mathcal{\mathring{L}}(0) + (\mathbb{W}\bm{\mathring{\psi}}_\star)^\prime \bm{\ell} \label{eq:D2approx}
\end{equation*}
as numerical approximation of the \textit{cyclical steady-state} ARL.
Note that being in the \textit{in-control} case, we receive $\Psi_0 \, \mathcal{\mathring{L}}(0) = 1$,
because $\mathcal{\mathring{L}}(0) = E_\infty(N)$.
In the here following Figure~\ref{fig:psi} we illustrate the shape of $\mathring{\psi}()$ and $\mathring{\psi}_\star()$
for dimensions $p\in\{2,3,4,10\}$ and $E_\infty(N) = 200$.
\begin{figure}[hbtp]
	
\subfloat[$p=2$]{\includegraphics[width=.51\textwidth]{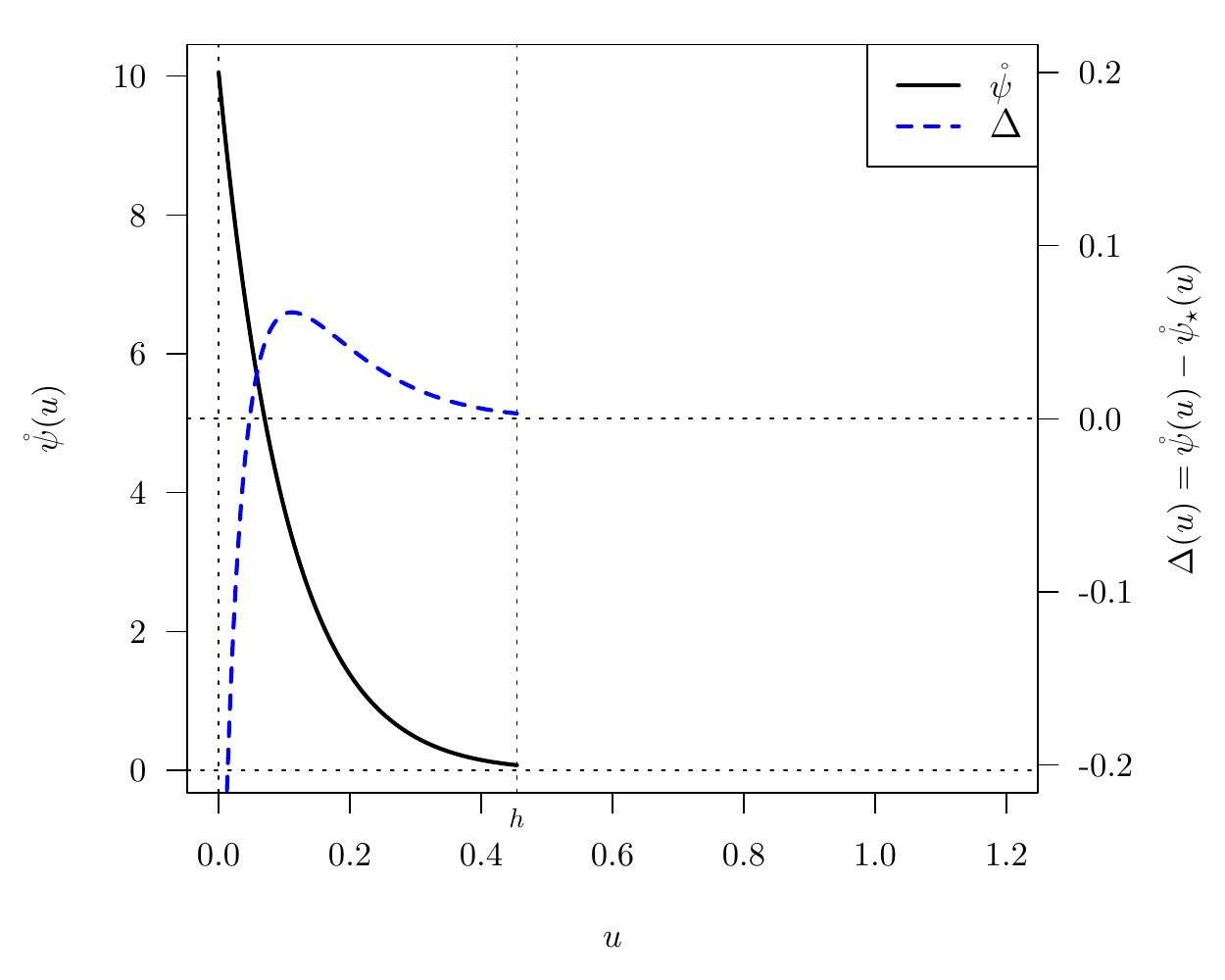}}
\subfloat[$p=3$]{\includegraphics[width=.51\textwidth]{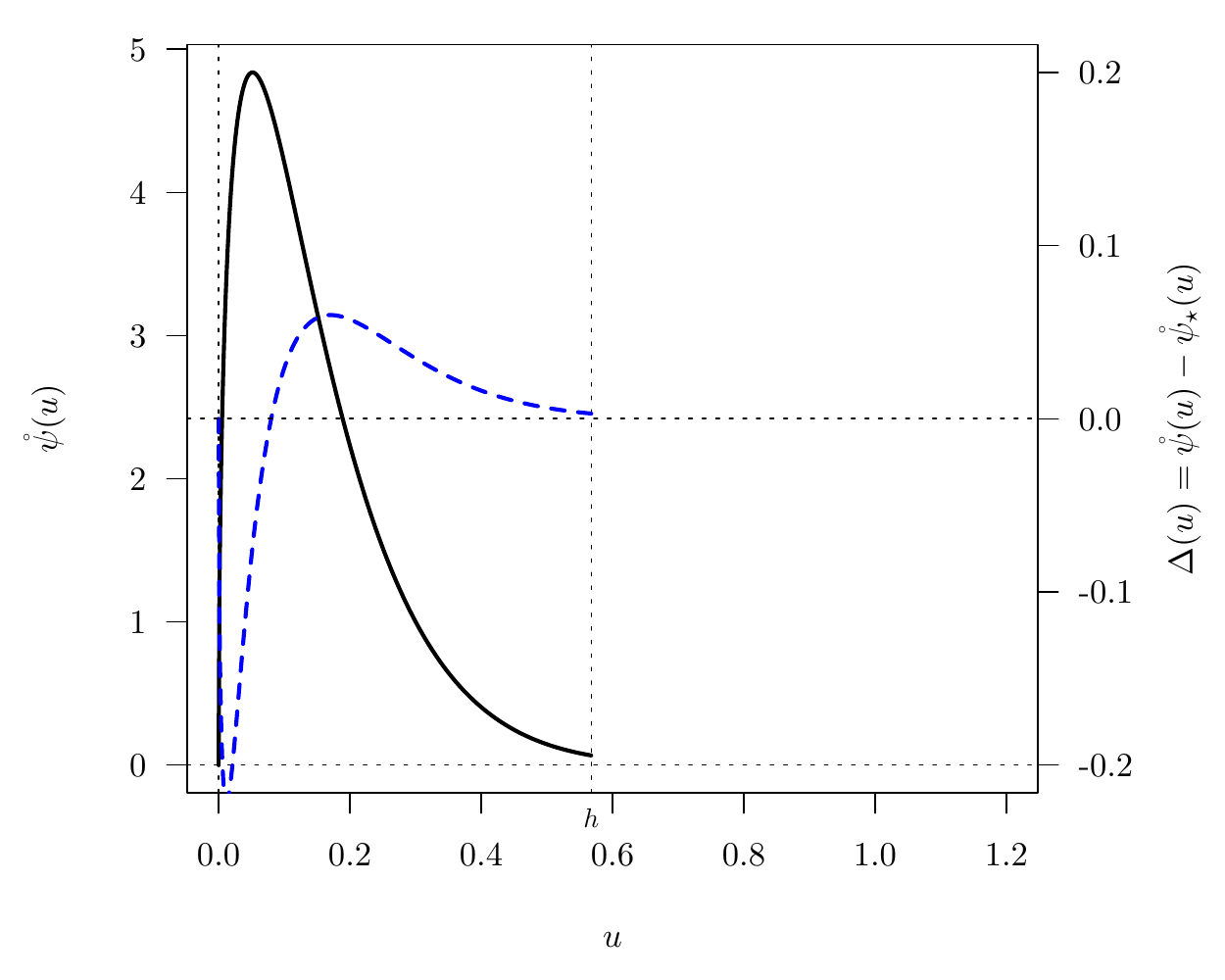}}
	
\subfloat[$p=4$]{\includegraphics[width=.51\textwidth]{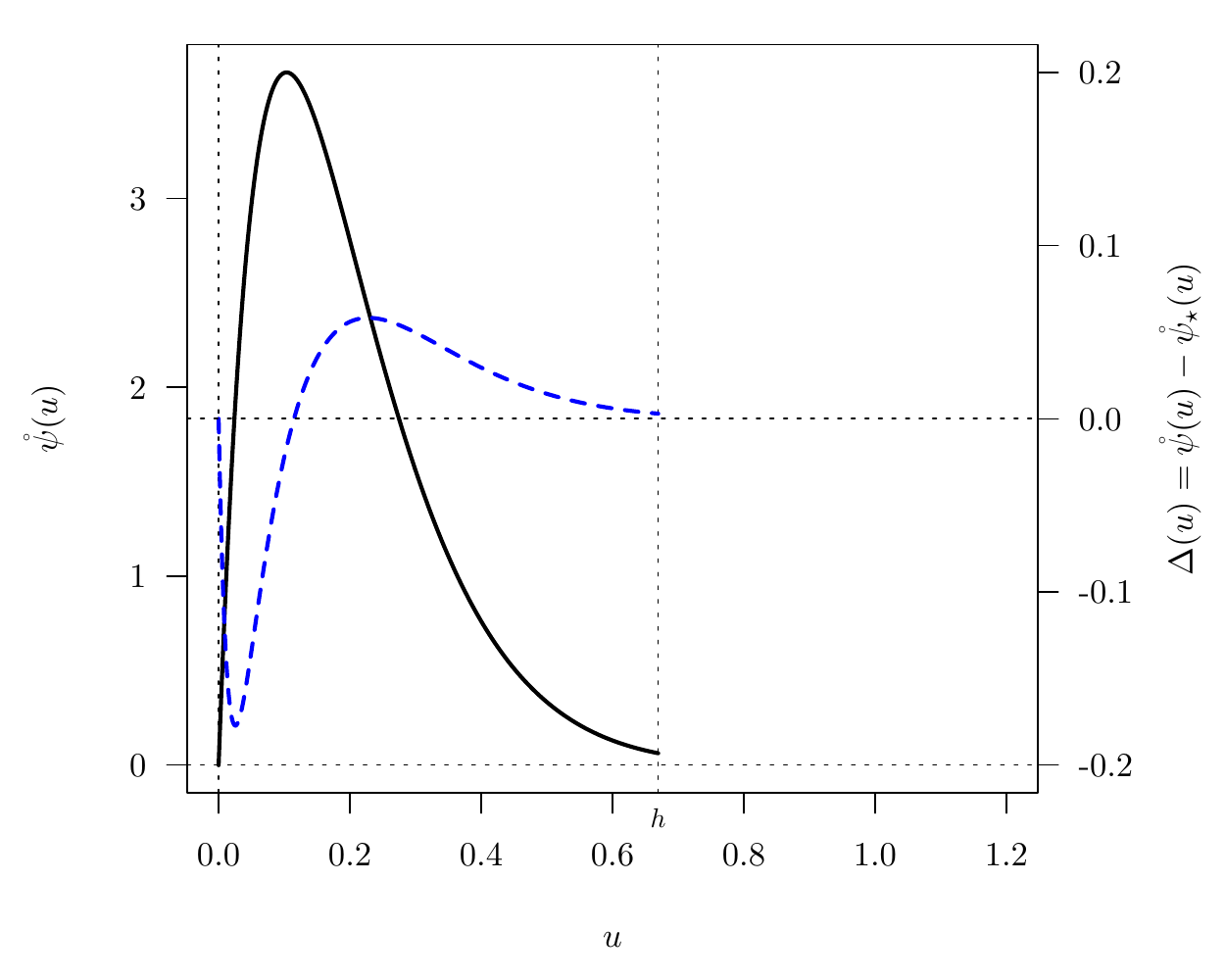}}
\subfloat[$p=10$]{\includegraphics[width=.51\textwidth]{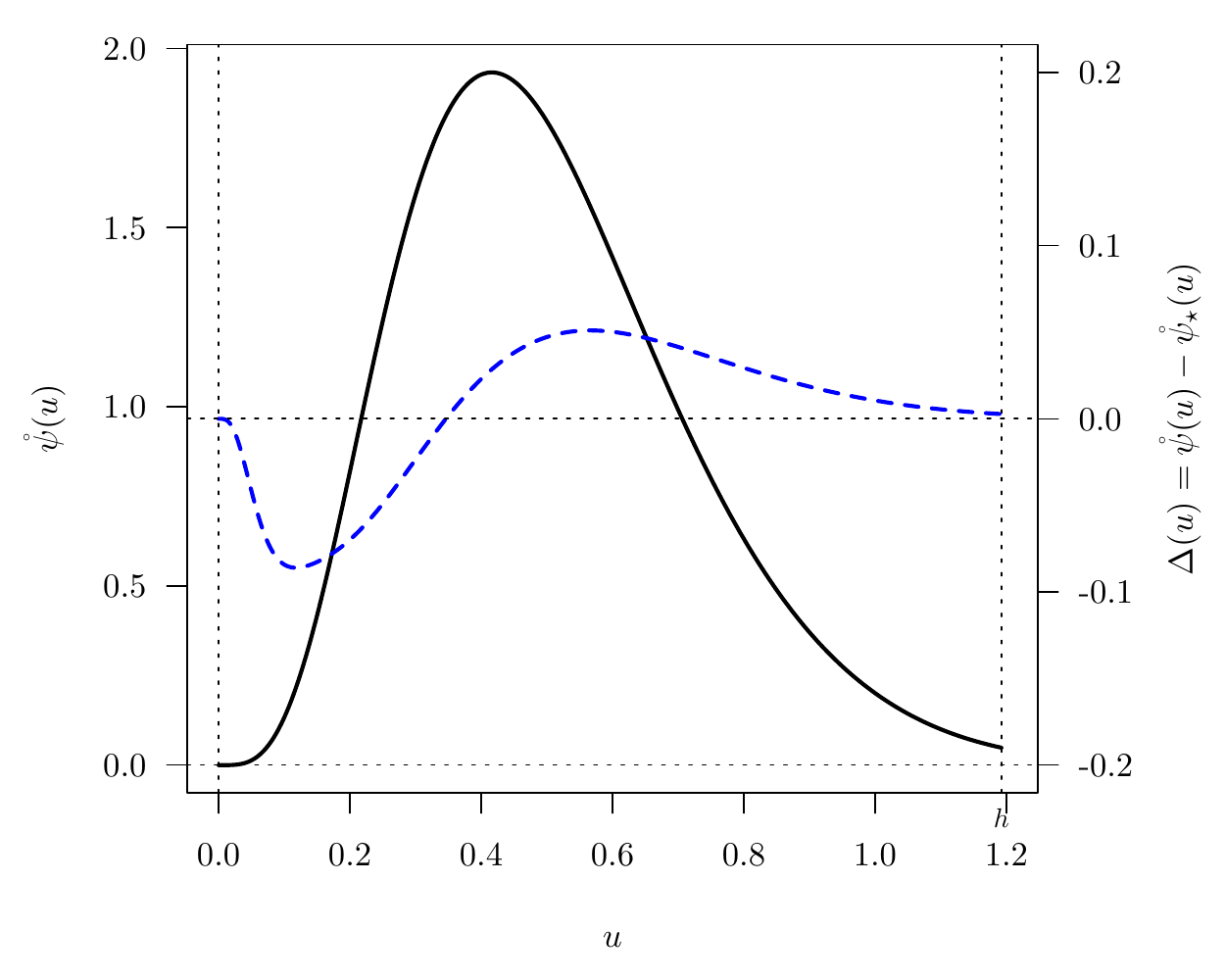}}
	
\caption{Left eigenfunctions, $\mathring{\psi}()$ and $\mathring{\psi}_\star()$
(actually, $\Delta() = \mathring{\psi}() - \mathring{\psi}_\star()$ is plotted),
needed to calculate both \textit{steady-state} ARL types, $\mathring{\mathcal{D}}$ and $\mathring{\mathcal{D}}_\star$, in the \textit{in-control} case
for selected dimensions $p$, smoothing constant $\lambda = 0.1$, and $E_\infty(N) = 200$.}\label{fig:psi}
\end{figure}
Because the differences between $\mathring{\psi}()$ and $\mathring{\psi}_\star()$ are quite small,
we plot $\mathring{\psi}()$ and $\Delta() = \mathring{\psi}() - \mathring{\psi}_\star()$.
The corresponding thresholds for $\lambda=0.1$ are $h_4 \in \{8.63, 10.78, 12.72, 22.66\}$.
The smallest dimension, $p=2$, yields a special shape of $\mathring{\psi}()$ which looks similar to an exponential distribution, while
for all other $p$ we observe $\mathring{\psi}(0) = 0$ and a pronounced mode at $u > 0$ which increases with $p$. 
Because the re-start for the \textit{cyclical} case takes place at $\bm{z}_0 = \bm{0}$ yielding $u = 0$,
$\mathring{\psi}_\star(u)$ is larger than  $\mathring{\psi}(u)$ for small $u$.
For larger values of $u$, $\mathring{\psi}(u)$ is mildly larger than $\mathring{\psi}_\star(u)$.
Note that the area under $\mathring{\psi}()$ is equal to 1, while it is $1 - \Psi_0 = 0.995$ for
$\mathring{\psi}_\star()$, because of the atom at $u = 0$.
The latter is the \textit{steady-state} probability of $\bm{Z}_{\tau-1} = \bm{0}$ which is equal to $\Psi_0 = 1/200$.
In addition, we provide some numerical results for $\mathring{\mathcal{D}}$ and $\mathring{\mathcal{D}}_\star$
for different $\lambda$ and $p$ in Table~\ref{tab:01}.
\begin{table}[hbt]
\centering
\small\renewcommand{\arraystretch}{1.4}
\caption{\textit{In-control steady-state} ARL results for various $p$,
smoothing constants $\lambda \in \{0.05, 0.1, 0.2\}$, and $E_\infty(N) = 200$.}\label{tab:01}	
\begin{tabular}{c|ccccccc}  \hline
  & \multicolumn{7}{|c}{$p$} \\ \hline
  ARL & 2 & 3 & 4 & 5 & 10 & 20 & 50 \\ \hline
  \multicolumn{8}{c}{$\lambda = 0.05$} \\ \hline
  $\mathring{\mathcal{D}}$         & 187.0 & 185.5 & 184.4 & 183.5 & 180.8 & 178.0 & 174.1 \\
  $\mathring{\mathcal{D}}_\star$   & 187.6 & 186.2 & 185.2 & 184.4 & 181.9 & 179.4 & 176.0 \\ \hline 
  \multicolumn{8}{c}{$\lambda = 0.1$} \\ \hline
  $\mathring{\mathcal{D}}$       & 192.6 & 191.8 & 191.3 & 190.9 & 189.6 & 188.3 & 186.4 \\
  $\mathring{\mathcal{D}}_\star$ & 192.7 & 192.1 & 191.6 & 191.2 & 190.0 & 188.7 & 186.9 \\ \hline
  \multicolumn{8}{c}{$\lambda = 0.2$} \\ \hline
  $\mathring{\mathcal{D}}$       & 196.1 & 195.8 & 195.6 & 195.4 & 194.8 & 194.2 & 193.3 \\
  $\mathring{\mathcal{D}}_\star$ & 196.2 & 195.9 & 195.7 & 195.5 & 194.9 & 194.3 & 193.5 \\  \hline
\end{tabular}
\end{table}
We notice that the \textit{steady-state} ARL gets closer to the \textit{zero-state} ARL for increasing
$\lambda$. Moreover, the differences between $\mathring{\mathcal{D}}$ and $\mathring{\mathcal{D}}_\star$ vanish simultaneously.
On the other hand, for increasing dimension $p$, the \textit{steady-state} ARL decreases.
Closing this section, we want to note two points. First, both types of the \textit{in-control steady-state}
ARL could be understood as expected number of observations until a false signal,
told to a person (either the control chart owner or just a ``witness'') who checks the state of a MEWMA
chart on an arbitrary day knowing only that ``no alarm occurred so far'' (\textit{conditional})
or ``only false alarms occurred'' (\textit{cyclical}). Second, \cite{Prab:Rung:1997} reported in their
Table 3 (rows with $\delta = 0$) corresponding numbers that are surprisingly close to 200, the \textit{zero-state} value.
The Monte Carlo studies we performed in Section~\ref{sec:comp} confirm essentially our numbers.
In the next section we turn to the more involved \textit{out-of-control} case.

\section{Analysis of the out-of-control case}\label{sec:num}

First, we recall the framework of \cite{Prab:Rung:1997}.
They claimed to \textit{``provide conditional steady-state ARLs in Table 3.''}
In order to do this, they took the transition matrix $\mathbb{P}_0$
(counterpart to our matrix $\mathbb{Q}_\mathcal{L}$, but $\mathbb{P}_0$ refers to a bivariate Markov chain),
calculated $\bm{v}^\prime = \bm{b}^\prime (\mathbb{I} - \mathbb{P}_0)^{-1}$ and created by $\bm{q} = \bm{v} / (\bm{v}^\prime \bm{1})$
the needed \textit{steady-state} distribution. Remind that the vector $\bm{b}$ consists of only zeros except for the component that
corresponds to the starting state ($\bm{z}_0$), which is set to 1.
Following \cite{Darr:Sene:1965}, \cite{Prab:Rung:1997} rather determined the \textit{cyclical steady-state} distribution, and consequently $\mathcal{D}_\star$.
Later on, we will see that this is, of course, mathematically important, but in terms of actual numbers less relevant.
Here we calculate both types, $\mathcal{D}$ and $\mathcal{D}_\star$, while connecting them carefully to their theoretical origins.
In addition, we propose an algorithm which avoids the problem of dealing with such large matrices like $\mathbb{P}_0$ which
stems from a bivariate Markov chain.

Contrary to \cite{Prab:Rung:1997}, we start from the (double) integral equation developed in \cite{Rigd:1995b} for the \textit{out-of-control zero-state} ARL,
which is a function of two arguments. Besides the already introduced $\alpha= \bm{z}_0^\prime \bm{z}_0$ we utilize
as second argument $\beta=\bm{\mu}_1^\prime \bm{z}_0$.
Recall also the value $\delta = \bm{\mu}_1^\prime \bm{\mu}_1$ which quantifies the magnitude of the change.
From \cite{Rigd:1995b} we present for the \textit{out-of-control zero-state} ARL ($\alpha\in[0,h]$, $\beta^2 \le \alpha \delta$)
\begin{align*}
  \mathcal{L}(\alpha,\beta) &
  = 1 + \int_{-\sqrt{\delta h}}^{\sqrt{\delta h}} \int_{v^2/\delta}^h \mathcal{L}(u,v) K(u,v;\alpha,\beta) \,\mathrm{d}u\,\mathrm{d}v  \,, \\
  K(u,v;\alpha,\beta) & = \frac{1}{\sqrt{2\pi\delta\lambda^2}}
  e^{-\frac{[v-\lambda\delta-(1-\lambda)\beta]^2}{2\delta\lambda^2}} \\
  & \qquad \times
  \frac{1}{\lambda^2} f_{\chi^2}\left(\!\frac{u-v^2/\delta}{\lambda^2} \,\Big|\, p-1, \eta (\alpha-\beta^2/\delta)\!\right) \,.
\end{align*}
Regarding the numerical solution of this double integral equation we refer to \cite{Rigd:1995b} and \cite{Knot:2017a}.
For preparing the left eigenfunction equation, we change the integration order and the second argument.
Writing $\theta$ for the angle between the initial MEWMA value $\bm{z}_0$ and the new mean vector $\bm{\mu}_1$, we
deduce $\beta = \sqrt{\alpha} \sqrt{\delta} \cos(\theta)$.
The new second argument is set by $\gamma = -\cos(\theta)$ for convenience (more details see the Appendix)
so that we get the following double integral equation:
\begin{align}
  \mathcal{L}(\alpha,\gamma) & = 1 + \int_0^h \int_{-1}^1 \mathcal{L}(u,w) K^\dag(u,w;\alpha,\gamma) \,\mathrm{d}w\,\mathrm{d}u \,, \label{eq:biLnew} \\
  K^\dag(u,w;\alpha,\gamma) & = \frac{\sqrt{u}}{\sqrt{2\pi\lambda^2}}
  e^{-\frac{[\sqrt{u} w - \lambda\sqrt{\delta} - (1-\lambda)\sqrt{\alpha} \gamma]^2}{2\lambda^2}} \nonumber \\
  & \qquad \times
  \frac{1}{\lambda^2} f_{\chi^2}\left(\frac{u(1-w^2)}{\lambda^2} \,\Big|\, p-1, \eta \alpha(1-\gamma^2) \right) \,. \nonumber
\end{align}  
Studying \eqref{eq:biLnew} we conclude that we have to account for the (quasi-)stationary distributions of both
the distance to zero, $\alpha$, and the angle $\theta$ ($\gamma = -\cos(\theta)$ links $\theta$ and $\gamma$) between the new mean $\bm{\mu}_1$ and the
MEWMA statistic $\bm{Z}_{\tau-1}$.
Similarly to the transition from $\mathcal{\mathring{L}}(\alpha)$ to $\mathring{\psi}(u)$ we derive the integral equation for $\psi(u,w)$,
\begin{equation}
  \varrho \psi(u,w) = \int_0^h \int_{-1}^{1} \psi(\alpha,\gamma) K^\dag(u,w;\alpha,\gamma) \,\mathrm{d}\gamma\,\mathrm{d}\alpha \,,
  \label{eq:psi2aigl}
\end{equation} 
and for the \textit{cyclical} case (restart at $\bm{z_0} = \bm{0}$)
\begin{equation}  
  \psi_\star(u,w) = \Psi_0 K^\dag(u,w;0,0)
  + \int_0^h \int_{-1}^{1} \psi_\star(\alpha,\gamma) K^\dag(u,w;\alpha,\gamma) \,\mathrm{d}\gamma\,\mathrm{d}\alpha \,.
  \label{eq:psi2bigl}
\end{equation}
Because we evaluate both eigenfunction
equations for the \textit{in-control} parameter setup analogously to \cite{Prab:Rung:1997}, we set $\delta=0$ and simplify
\begin{align*}  
  K^\dag(u,w;\alpha,\gamma) & = \ldots e^{-\frac{[\sqrt{u} w - (1-\lambda)\sqrt{\alpha} \gamma]^2}{2\lambda^2}}  \ldots \,, \\
  K^\dag(u,w;0,0)
  & = \frac{\Gamma(\frac{p}{2})}{\Gamma(\frac{p-1}{2}) \sqrt{\pi}} (1-w^2)^{\frac{p-3}{2}} \times \frac{1}{\lambda^2} f_{\chi^2}\left(\frac{u}{\lambda^2} \,\Big|\, p \right) \,.
\end{align*}
In the sequel, we will demonstrate that the solutions of both integral equations are degenerated.
Interestingly, the second term of $K^\dag(u,w;0,0)$ coincides with the factor at $\Psi_0$ in the integral equation \eqref{eq:psi0bigl}.

Now we merely assume that $\psi(u,\gamma) = \mathring{\psi}(u) \times d(\gamma)$ where $d(\gamma)$ will be defined below.
Moreover, we heuristically proceed and conjecture that the projection of $\bm{Z}_{\tau-1}$ to the unit sphere 
$S^{p-1} = \{ \bm{x} \in \mathds{R}^p\!: \Vert \bm{x}\Vert = 1\}$ results into a uniform distribution on $S^{p-1}$.
Given a spherical distribution (multivariate normal is one prominent example), it would be a well-known result, see, e.\,g., \cite{Muir:1982}, Theorem 1.5.6.
The unrestricted (neither conditioning on $N \ge \tau$ nor re-starting after false alarm) sequence $\bm{Z}_i$
follows a multivariate normal distribution with mean $\bm{0}$ and covariance matrix $\lambda\big(1-(1-\lambda)^{2i}\big)/(2-\lambda) \mathbb{I}$.
Within this framework, there exists a beneficial result for the distribution of the angle $\theta$ between one fixed and one uniformly
chosen point or two uniformly chosen points on the unit sphere. Following \cite{Muir:1982}, Theorem 1.5.5 we conclude for the
density of $\theta$
\begin{equation*}
  d(\theta) = \frac{\Gamma(\frac{p}{2})}{\Gamma(\frac{p-1}{2}) \sqrt{\pi}} \, \sin(\theta)^{p-2}
\end{equation*}
with the special case $d(\theta) = 1/\pi$ for $p=2$.
A simple sketch of proof is given in the Appendix.
Rewriting $d()$ as function of $\gamma = - \cos(\theta)$ yields 
\begin{equation}
  d(\gamma) = \frac{\Gamma(\frac{p}{2})}{\Gamma(\frac{p-1}{2}) \sqrt{\pi}} \, (1-\gamma^2)^{\frac{p-3}{2}}
  = \frac{(1-\gamma^2)^{\frac{p-3}{2}}}{B\big(1/2, (p-1)/2\big)} 
  \qquad,\;\; \gamma \in (-1,1] \label{eq:law:gamma}
\end{equation}
including the special case $d(\gamma) = 1/2$ for $p = 3$. Note that the square of $\gamma$ follows
a beta distribution with parameters $1/2$ and $(p-1)/2$ (see the denominator of the second ratio which represents the corresponding beta function).

Collecting the results we achieved so far, we formulate the following lemma.

\begin{lemma}\label{LEM:01}
For a MEMWA chart following \eqref{eq:mewma1} and \eqref{eq:N}, the solution $\psi(u,w)$ of \eqref{eq:psi2aigl} could be established
by combining $\mathring{\psi}(u)$ from \eqref{eq:psi0aigl} and $d(\gamma)$ in \eqref{eq:law:gamma}:
\begin{equation}
  \psi(u,w) = d(w) \times \mathring{\psi}(u) \,. \label{eq:paira}
\end{equation}
\end{lemma}

\begin{proof}
We start with inserting \eqref{eq:paira} into \eqref{eq:psi2aigl} on the right-hand side.
\begin{align*}
  \varrho \psi(u,w)
  & = \int_0^h \mathring{\psi}(\alpha) \int_{-1}^{1} d(\gamma) K^\dag(u,w;\alpha,\gamma) \,\mathrm{d}\gamma\,\mathrm{d}\alpha \\
  & = \frac{\Gamma(\frac{p}{2})}{\Gamma(\frac{p-1}{2})\sqrt{\pi}} \int_0^h  \mathring{\psi}(\alpha)
    \int_{-1}^{1} (1-\gamma^2)^{\frac{p-3}{2}} K^\dag(u,w;\alpha,\gamma) \,\mathrm{d}\gamma\,\mathrm{d}\alpha \,.
\end{align*}
In the Appendix we prove by utilizing a representation of the non-central $\chi^2$ density including
the confluent hypergeometric limit function ${_0} F_1()$ that
\begin{align*}
  \int_{-1}^{1} (1-\gamma^2)^{\frac{p-3}{2}} K^\dag(u,w;\alpha,\gamma) \,\mathrm{d}\gamma
  & = (1-w^2)^{\frac{p-3}{2}} \frac{1}{\lambda^2} f_{\chi^2}\left(\frac{u}{\lambda^2} \,\Big|\, p, \eta \alpha\right) \,.
\end{align*}
Deploying this result simplifies the above double integral to 
\begin{align*}
 \ldots
 & = \frac{\Gamma(\frac{p}{2})}{\Gamma(\frac{p-1}{2})\sqrt{\pi}} (1-w^2)^{\frac{p-3}{2}}
   \int_0^h \mathring{\psi}(\alpha) \frac{1}{\lambda^2} f_{\chi^2}\left(\frac{u}{\lambda^2} \,\Big|\, p, \eta \alpha\right) \, \mathrm{d}\alpha \\
 & = d(w) \times \varrho \mathring{\psi}(u)
\end{align*}
because of \eqref{eq:law:gamma} and \eqref{eq:psi0aigl}. Thus, \eqref{eq:paira} solves \eqref{eq:psi2aigl}.
\end{proof}

\noindent
In the same way we derive the left eigenfunction for the \textit{cyclical} case.

\begin{corollary}
For a MEMWA chart following \eqref{eq:mewma1} and \eqref{eq:N}, the solution $\psi_\star(u,w)$ of \eqref{eq:psi2bigl} 
is set up similarly to $\psi(u,w)$ by combining $\mathring{\psi}_\star(u)$ from \eqref{eq:psi0bigl} and $d(\gamma)$ in \eqref{eq:law:gamma}:
\begin{equation}
  \psi_\star(u,w) = d(w) \times \mathring{\psi}_\star(u) \,. \label{eq:pairb}
\end{equation}
\end{corollary}

\begin{proof}
Inserting \eqref{eq:pairb} into \eqref{eq:psi2bigl} on the right-hand side.
\begin{align*}
	\lefteqn{\psi_\star(u,w)} \\	
	& = \Psi_0 \, d(w) \, \frac{1}{\lambda^2} f_{\chi^2}\left(\frac{u}{\lambda^2} \,\Big|\, p \right)	  \\
	& \qquad    
	  + \frac{\Gamma(\frac{p}{2})}{\Gamma(\frac{p-1}{2})\sqrt{\pi}} \int_0^h  \mathring{\psi}_\star(\alpha)
   	  \int_{-1}^{1} (1-\gamma^2)^{\frac{p-3}{2}} K^\dag(u,w;\alpha,\gamma) \,\mathrm{d}\gamma\,\mathrm{d}\alpha \\
   	& = \Psi_0 \, d(w) \, \frac{1}{\lambda^2} f_{\chi^2}\left(\frac{u}{\lambda^2} \,\Big|\, p \right)
   	  + d(w) \int_0^h \mathring{\psi}_\star(\alpha) \frac{1}{\lambda^2} f_{\chi^2}\left(\frac{u}{\lambda^2} \,\Big|\, p, \eta \alpha\right) \, \mathrm{d}\alpha \\
    & = d(w) \times \mathring{\psi}_\star(u) \,.
\end{align*}
Again we use the \textit{in-control} integral equation, now \eqref{eq:psi0bigl} for the \textit{cyclical} case.
\end{proof}

Having the two left eigenfunctions in the shape we need for calculating the \textit{out-of-control steady-state} ARL, we
combine them with the ARL function which is determined through \eqref{eq:biLnew}:
\begin{align*}
  \mathcal{D} & = \int_0^h \mathring{\psi}(\alpha) \int_{-1}^1 d(\gamma) \mathcal{L}(\alpha,\gamma) \,\mathrm{d}\gamma\,\mathrm{d}\alpha \,, \\
  \mathcal{D}_\star & = \int_0^h \mathring{\psi}_\star(\alpha) \int_{-1}^1 d(\gamma) \mathcal{L}(\alpha, \gamma) \,\mathrm{d}\gamma\,\mathrm{d}\alpha
  + \Psi_0 \, \mathcal{L}(0,0) \,.
\end{align*}
Utilizing the node structure of the numerical solution of \eqref{eq:biLnew} in \cite{Knot:2017a}, we replace the above two double integrals by
quadrature and calculate, eventually, the two \textit{steady-state} ARL versions by evaluating the resulting double sum.
In the next section, the new numerical algorithms are used to produce maps of the left eigenfunction $\psi()$
(again, $\psi_\star()$ looks similarly)
and to calculate $\mathcal{D}$ and $\mathcal{D}_\star$ for comparison with, e.\,g., \cite{Prab:Rung:1997} and Monte Carlo results.
In addition, the \textit{worst-case} ARL is evaluated.

\section{Comparison studies}\label{sec:comp}

By using the algorithms presented in the previous sections, we want to illustrate the specific
shapes of the two functions constituting the integrand of the $\mathcal{D}$ integral.
To do this, we plot, in Figures~\ref{fig:map1} and \ref{fig:map2}, isolines of
$\mathcal{L}(\alpha, \theta)$ and the left eigenfunction $\psi(\alpha, \theta)$
where $\alpha = \bm{z}_0^\prime \bm{z}_0$ denotes the distance from $(0,0)$,
and $\theta = -\text{acos}(\gamma)$ is both the angle between the abscissa and the drawn vector
(see Figure~\ref{fig:map1}(b) and \ref{fig:map2}(f) for an example vector), and between $\bm{z}_0$ and $\bm{\mu}_1$.
In order to judge the usefulness of the \textit{zero-state} ARL, we added the respective isoline level.
\begin{figure}[hbtp]
\subfloat[$p=2$, $\mathcal{L}()$]{\includegraphics[width=.51\textwidth]{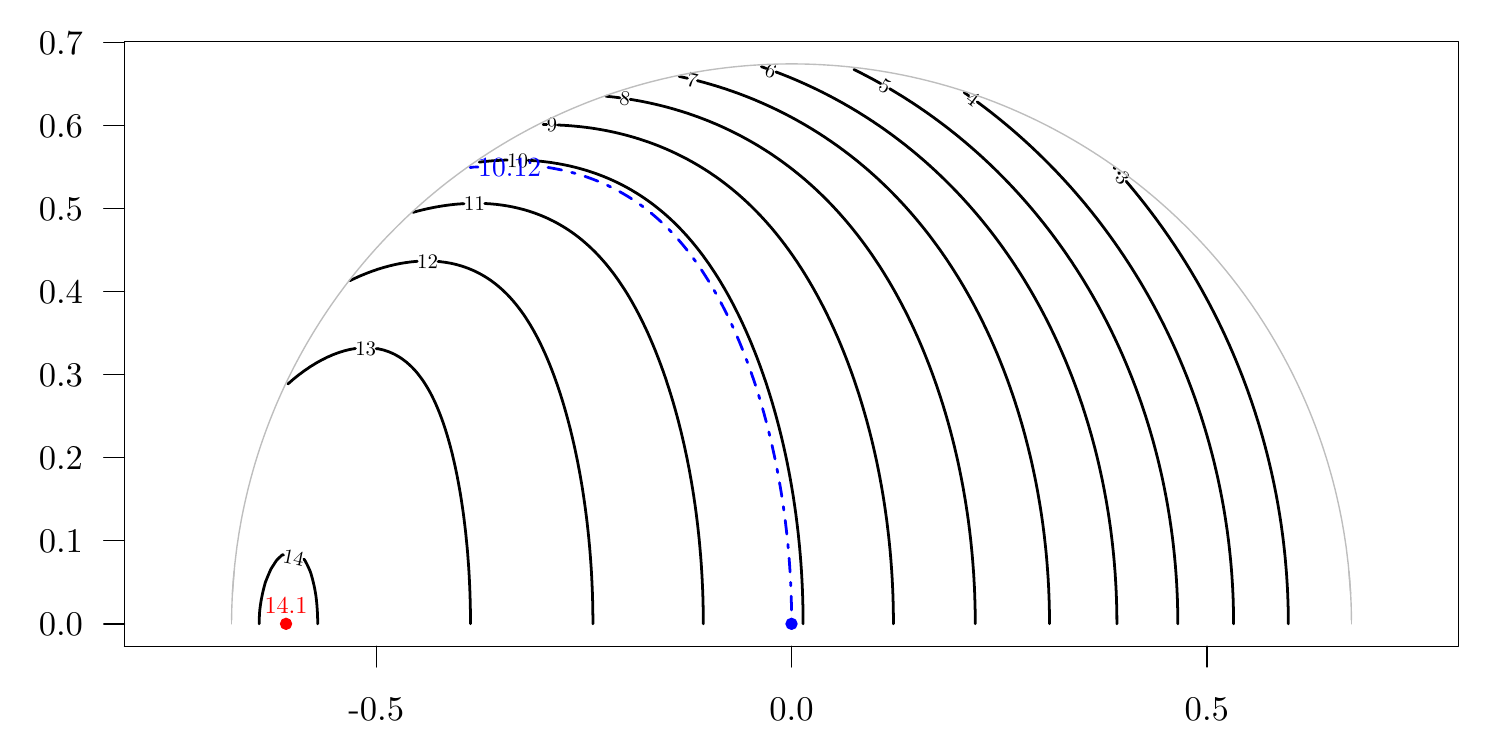}}
\subfloat[$p=2$, $\psi()$]{\includegraphics[width=.51\textwidth]{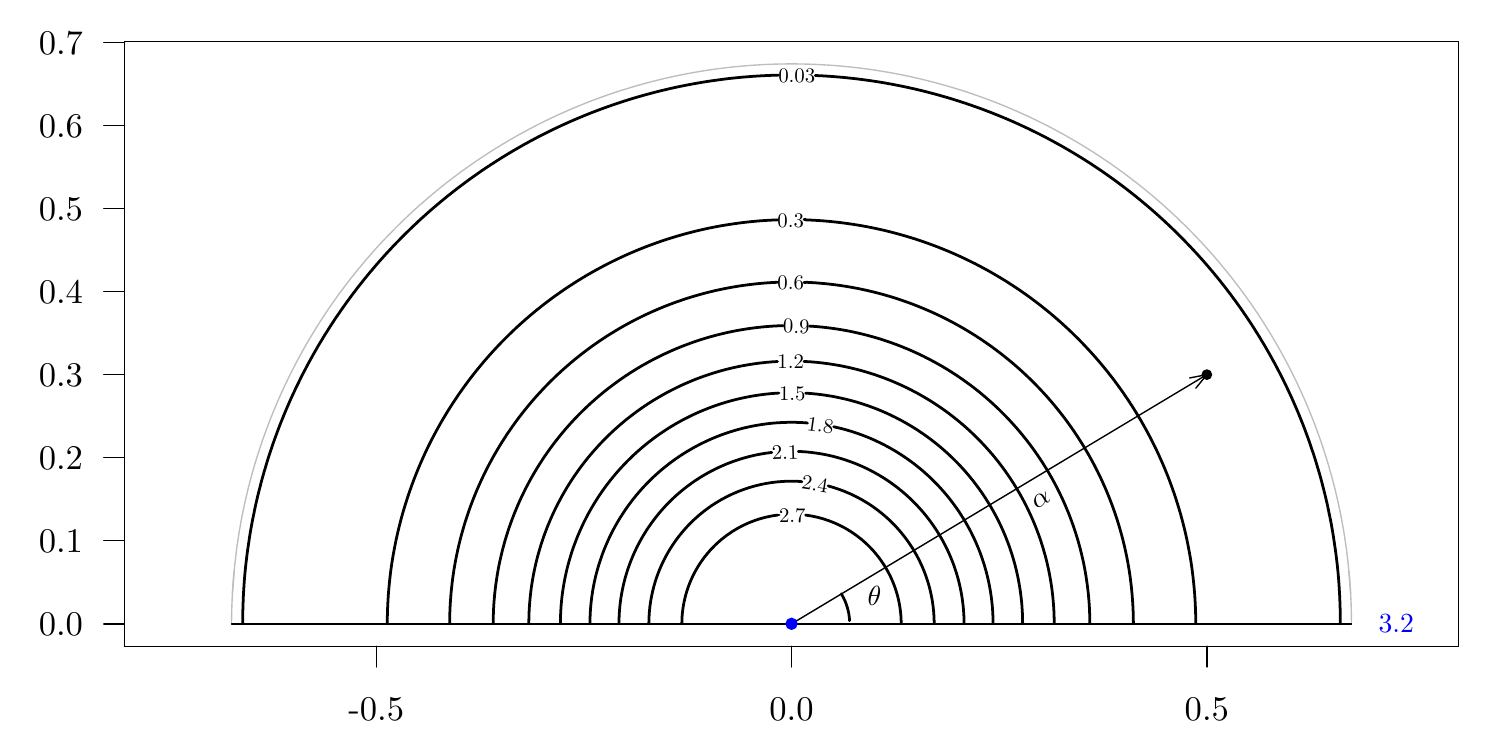}}
	
\subfloat[$p=3$, $\mathcal{L}()$]{\includegraphics[width=.51\textwidth]{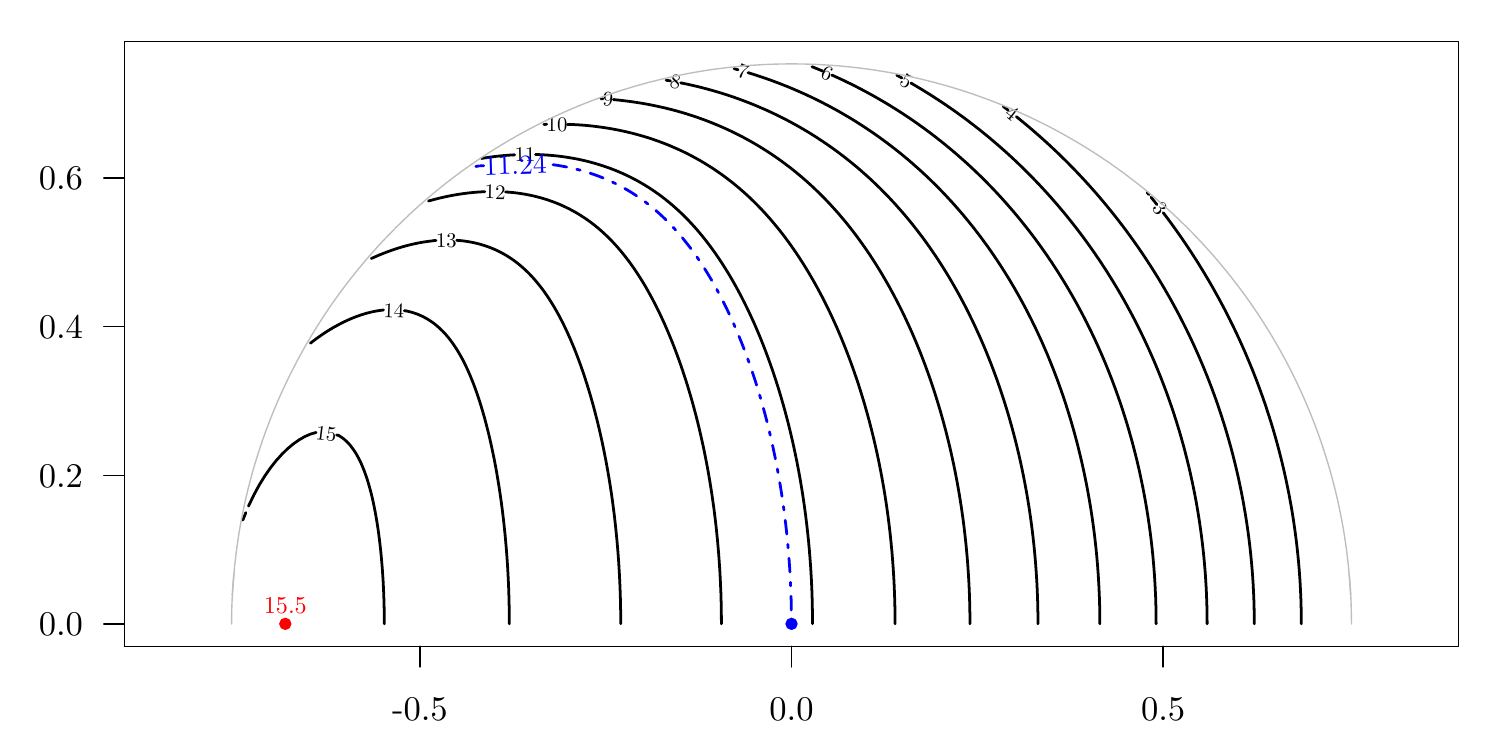}}
\subfloat[$p=3$, $\psi()$]{\includegraphics[width=.51\textwidth]{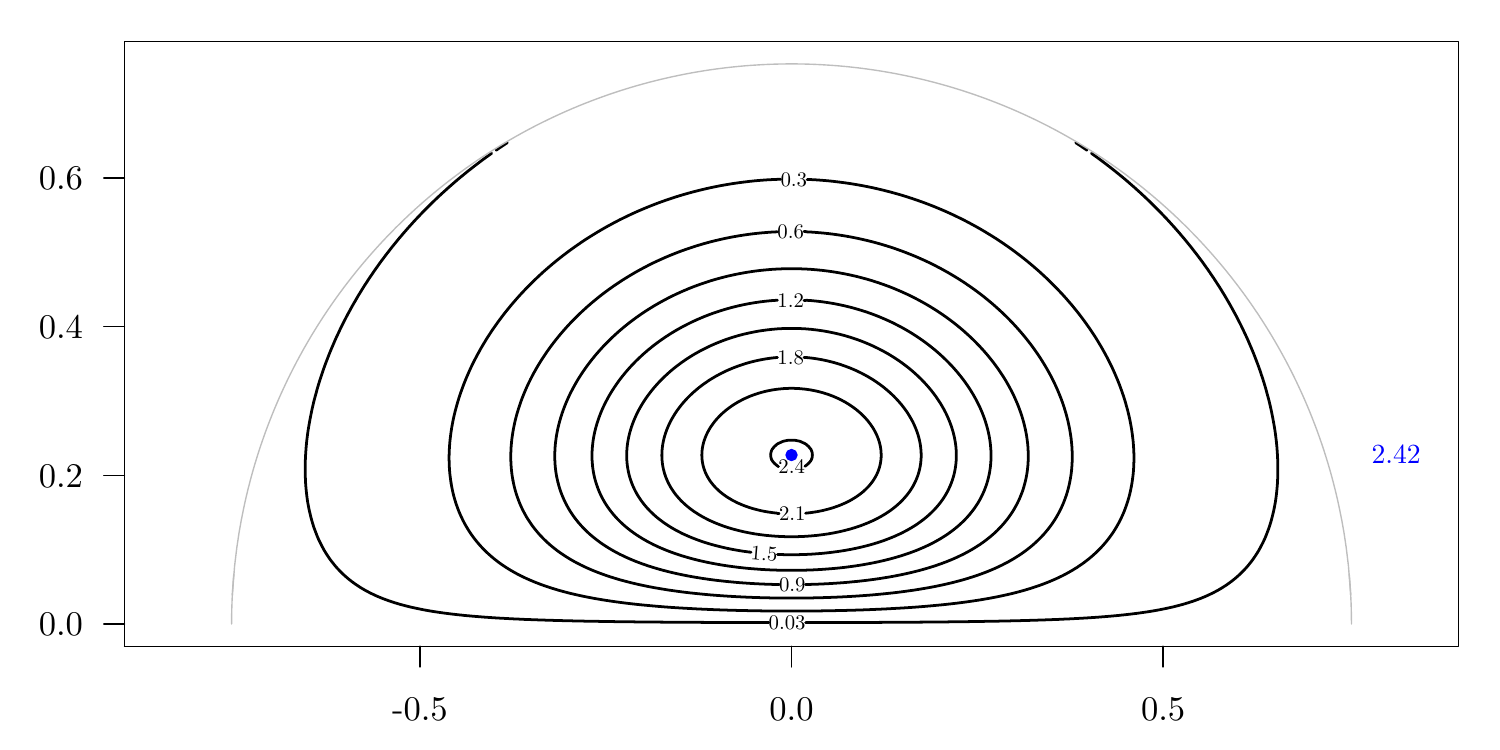}}
	
\subfloat[$p=4$, $\mathcal{L}()$]{\includegraphics[width=.51\textwidth]{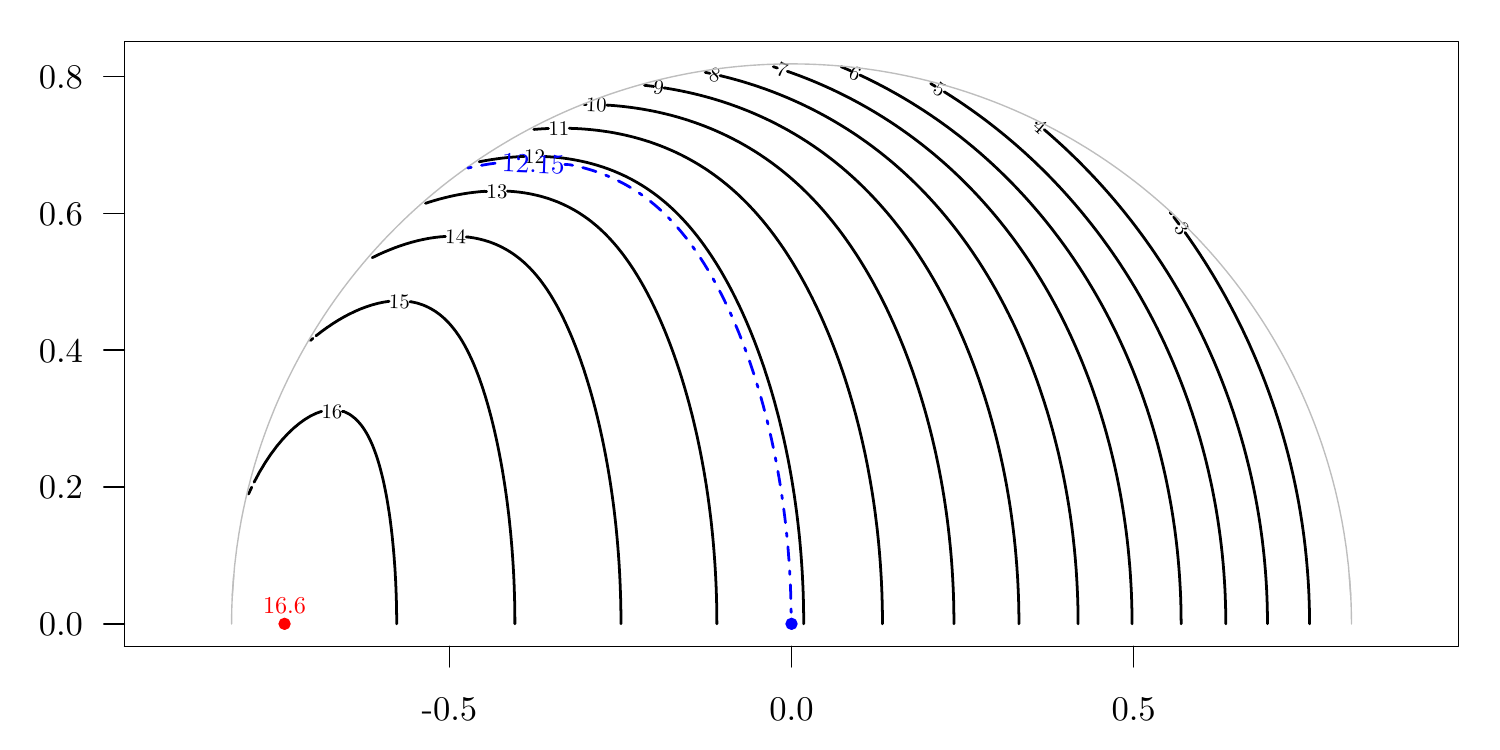}}
\subfloat[$p=4$, $\psi()$]{\includegraphics[width=.51\textwidth]{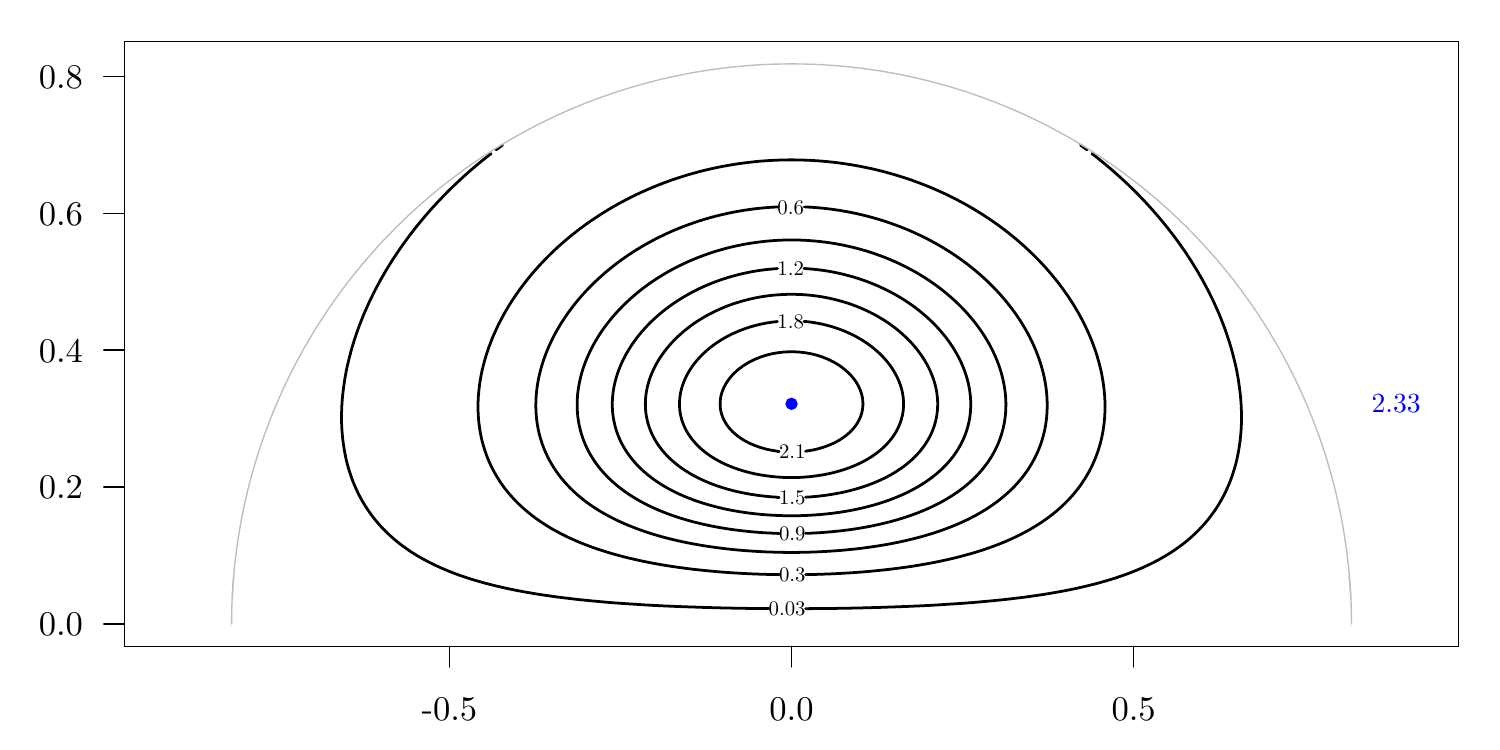}}
	
\caption{Isolines  of the ARL function $\mathcal{L}(\alpha, \theta)$
and the left eigenfunction $\psi(\alpha, \theta)$; polar type plot with $\alpha$ being
the distance of $\bm{Z}$ to zero, and $\theta$ exhibits the angle (latitude) between $\bm{Z}$ and the new mean $\bm{\mu}_1$ with $\Vert \bm{\mu}_1 \Vert = 1$;
$p\in\{2,3,4\}$.}\label{fig:map1}
\end{figure}
\begin{figure}[hbtp]
\subfloat[$p=5$, $\mathcal{L}()$]{\includegraphics[width=.51\textwidth]{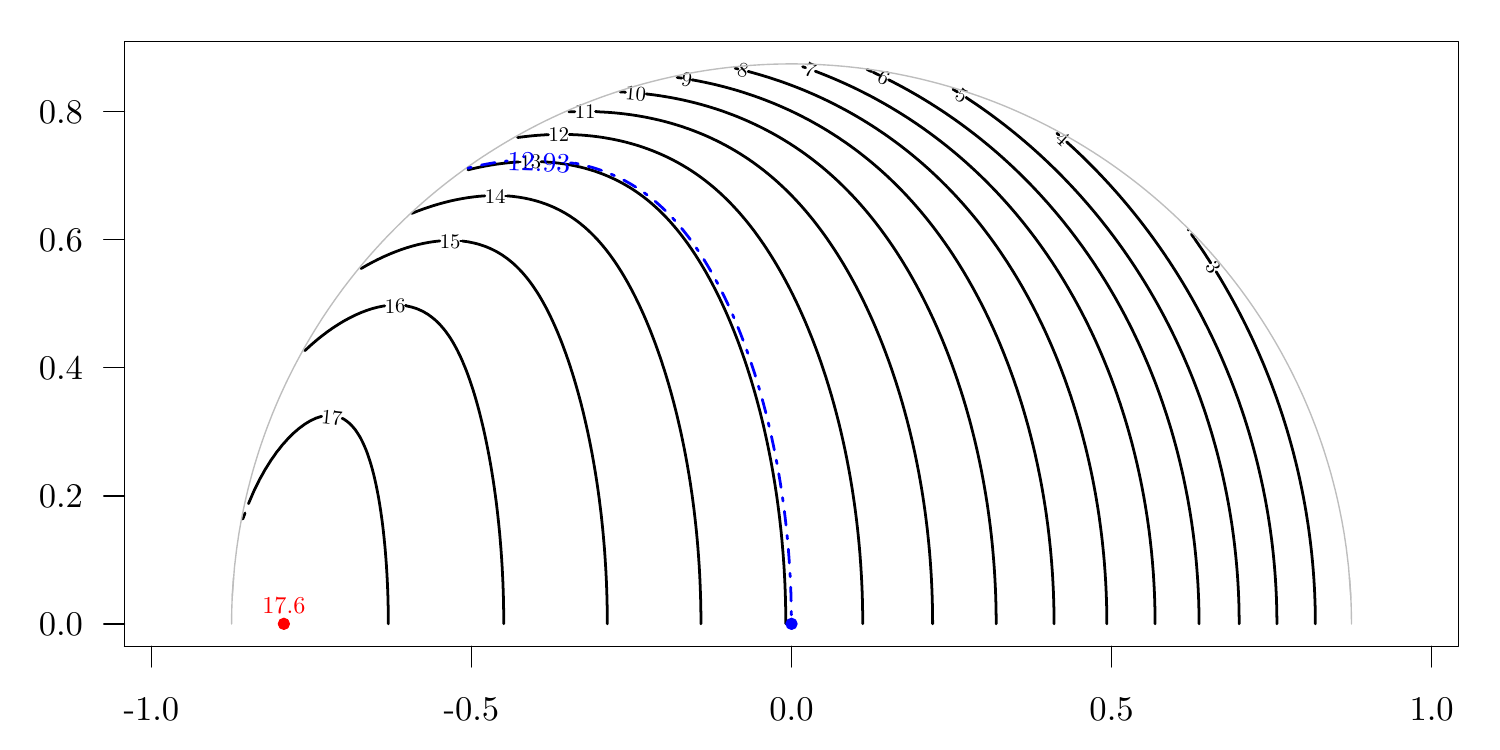}}
\subfloat[$p=5$, $\psi()$]{\includegraphics[width=.51\textwidth]{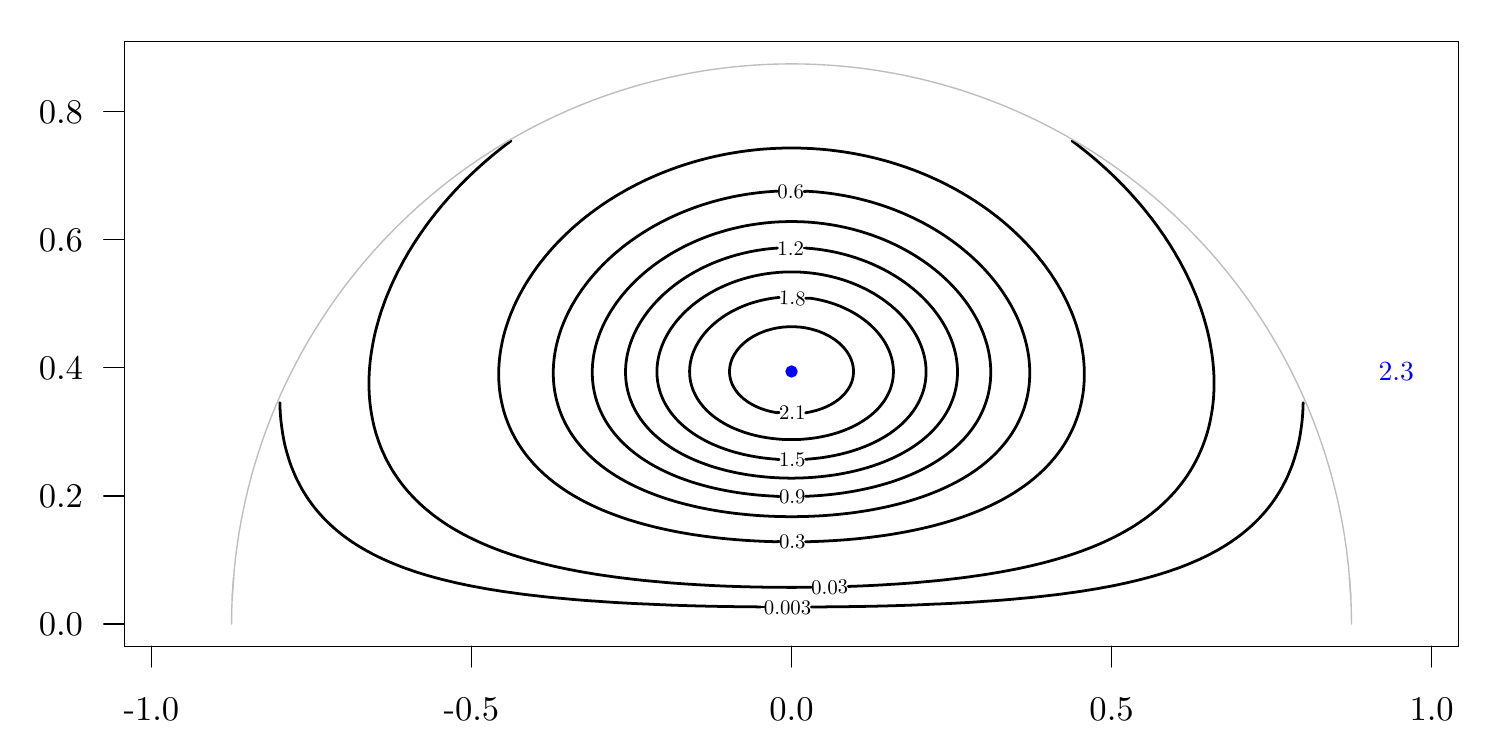}}
	
\subfloat[$p=10$, $\mathcal{L}()$]{\includegraphics[width=.51\textwidth]{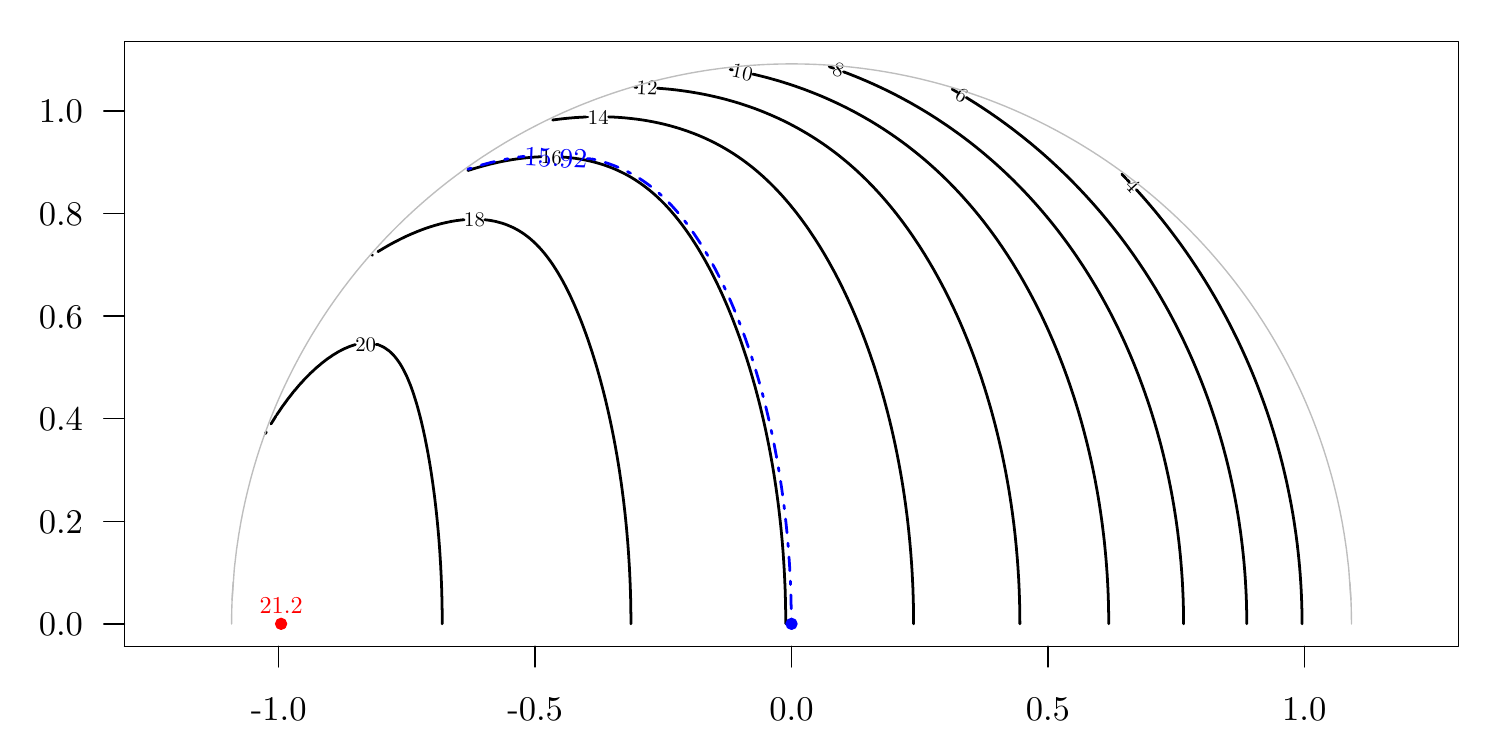}}
\subfloat[$p=10$, $\psi()$]{\includegraphics[width=.51\textwidth]{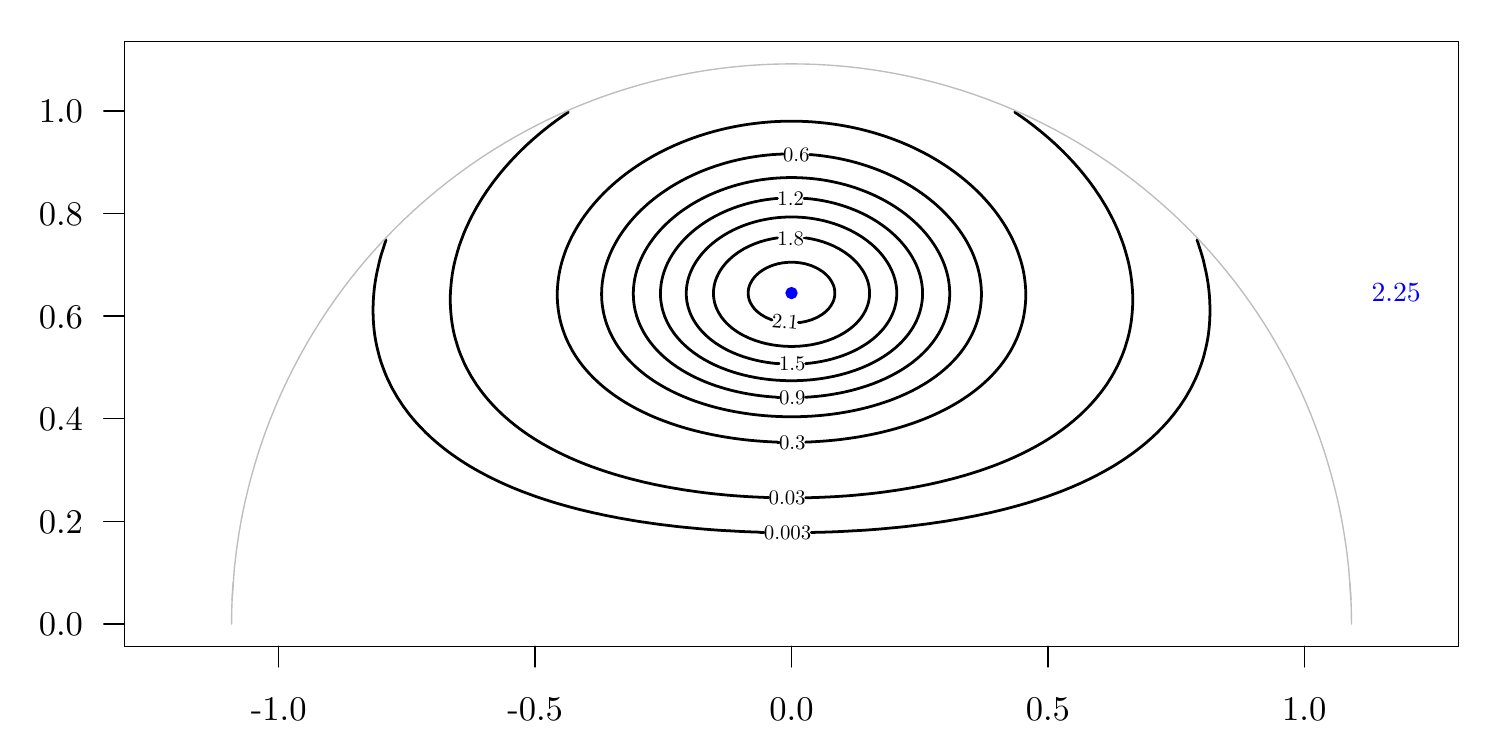}}
	
\subfloat[$p=20$, $\mathcal{L}()$]{\includegraphics[width=.51\textwidth]{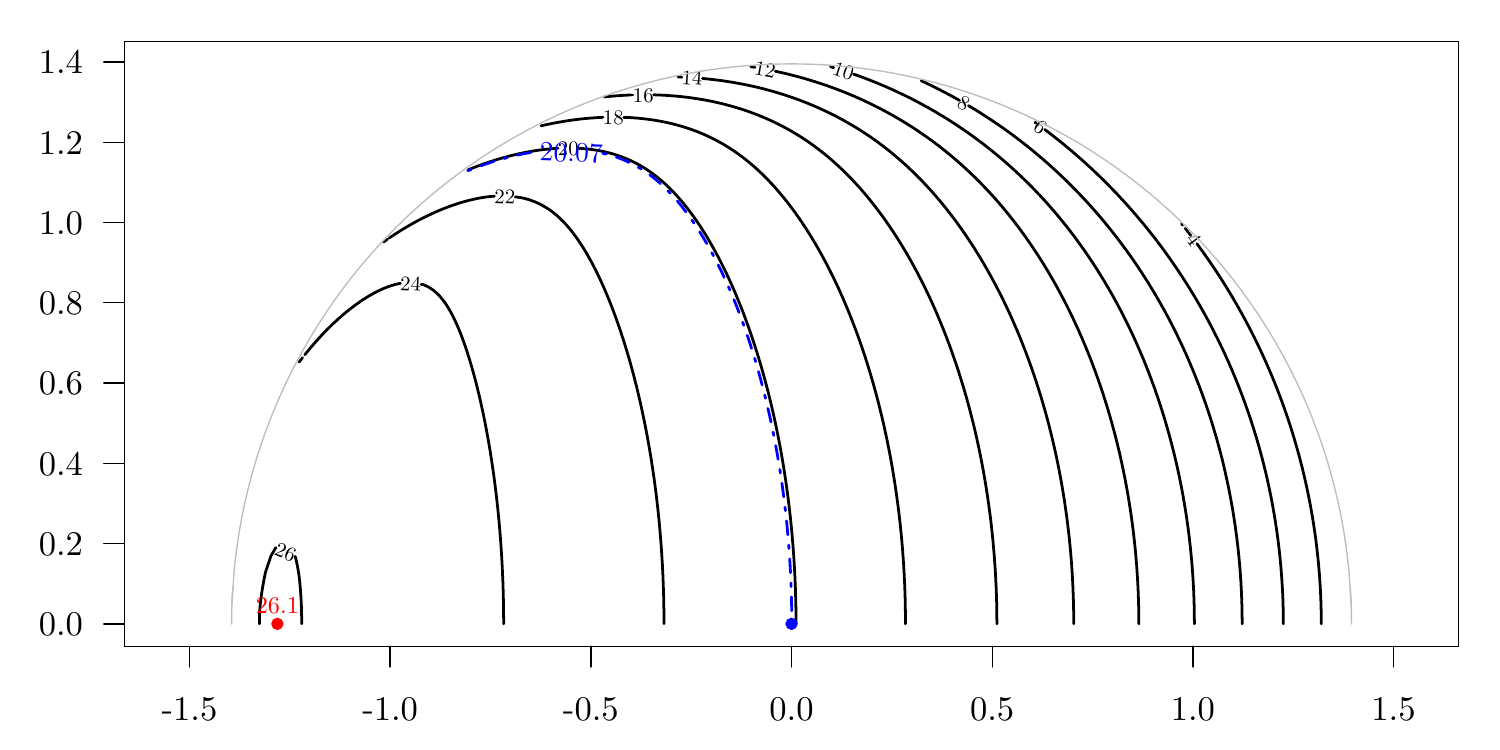}}
\subfloat[$p=20$, $\psi()$]{\includegraphics[width=.51\textwidth]{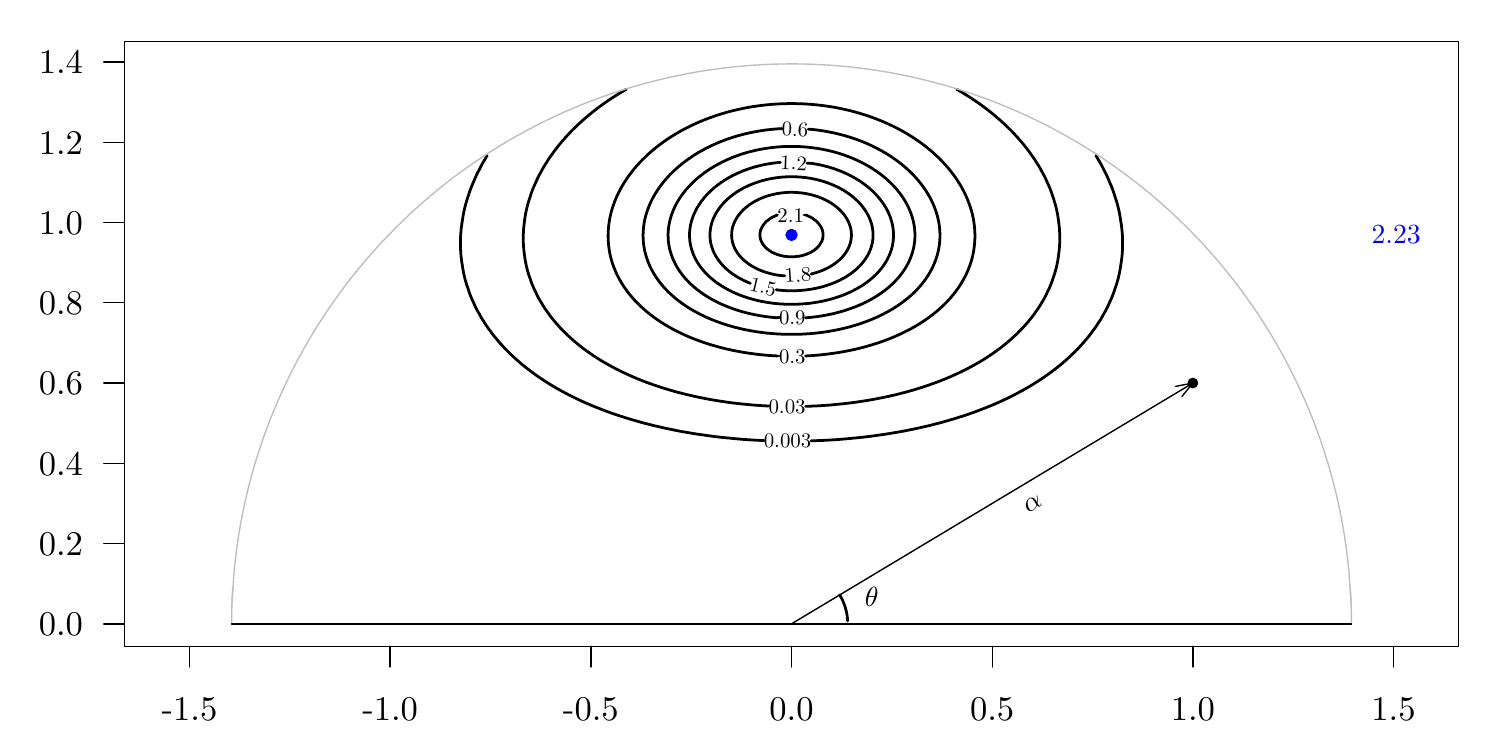}}
	
\caption{Isolines  of the ARL function $\mathcal{L}(\alpha, \theta)$ and the left eigenfunction $\psi(\alpha, \theta)$; polar type plot with $\alpha$ being
the distance of $\bm{Z}$ to zero, and $\theta$ exhibits the angle (latitude) between $\bm{Z}$ and the new mean $\bm{\mu}_1$ with $\Vert \bm{\mu}_1 \Vert = 1$;
$p\in\{5,10,20\}$.}\label{fig:map2}
\end{figure}
The six ARL plots look similarly.
Even the single point labeling the maximum \textit{out-of-control} ARL yields a nearly constant (relative to the threshold) position.
The maps of $\psi()$ start with circular isolines for $p = 2$
(remember the ``exponential'' shape of $\mathring{\psi}()$ and $d(\theta) = 1/\pi$). Then the probability mass wanders along
the $\theta = \pi/2$ line towards the  border. Overlaying the two maps we conclude that approaching the \textit{worst-case}
becomes less likely for increasing dimension.
In consequence, typical detection delays will be close to the \textit{steady-state} ARL, $\mathcal{D}$, which itself differs
not much from the \textit{zero-state} ARL.

Now, we want to compare results of all three ARL types for various values of $\delta = \bm{\mu}^\prime \bm{\mu}$ (and not only one
particular change $\bm{\mu}_1$).
In order to do so, we take from \cite{Prab:Rung:1997} several configurations for $\lambda = 0.1$ and $E_\infty(N) = 200$.
Recall that these authors utilized a bivariate Markov chain to approximate the stationary distribution
and the ARL function $\mathcal{L}()$. The corresponding matrix exhibits dimensions from 1\,500 ($m=30$) up to 6\,000 ($m=60$).
\begin{table}[hbtp]
\centering
\caption{ARL for $\lambda = 0.1$ and $p \in \{2, 3, 4, 10\}$; PR1997 refers to \cite{Prab:Rung:1997},
K2017 to \cite{Knot:2017a}, and MC to Monte Carlo ($10^9$ replications); \faUser\ --- new approach.}\label{tab:02}
\small
\begin{tabular}{c|*{3}{r}|*{5}{r}} \hline
  & \multicolumn{3}{c|}{\textit{zero-state}} & \multicolumn{5}{c}{\textit{steady-state}} \\
  $\sqrt{\delta}$ & PR1997 & K2017 & MC & PR1997 &
  $\mathcal{D}$, \faUser & $\mathcal{D}$, MC & $\mathcal{D}_\star$, \faUser & $\mathcal{D}_\star$, MC \\ \hline 
  \multicolumn{9}{c}{$p=2$, $h_4 = 8.64$} \\ \hline
  0   & 199.98 & 200.54 & 200.54   & 200.03 & 193.09 & 193.09 & 193.29 & 193.29 \\
  0.5 &  28.07 &  28.02 &  28.02   &  26.87 &  26.79 &  26.79 &  26.82 &  26.82 \\
  1   &  10.15 &  10.13 &  10.13   &   9.71 &   9.68 &   9.68 &   9.69 &   9.69 \\
  1.5 &   6.11 &   6.09 &   6.09   &   5.85 &   5.83 &   5.83 &   5.84 &   5.84 \\
  2   &   4.42 &   4.41 &   4.41   &   4.23 &   4.22 &   4.22 &   4.23 &   4.23 \\
  3   &   2.93 &   2.92 &   2.92   &   2.81 &   2.81 &   2.80 &   2.81 &   2.81 \\ \hline
  \multicolumn{9}{c}{$p=3$, $h_4 = 10.784$} \\ \hline
  0   &     -- & 200.03 & 200.03   &     -- & 191.86 & 191.86 & 192.09 & 192.09 \\
  0.5 &     -- &  31.85 &  31.85   &     -- &  30.22 &  30.22 &  30.26 &  30.26 \\
  1   &     -- &  11.24 &  11.24   &     -- &  10.60 &  10.60 &  10.62 &  10.62 \\
  1.5 &     -- &   6.71 &   6.71   &     -- &   6.31 &   6.31 &   6.32 &   6.32 \\
  2   &     -- &   4.83 &   4.83   &     -- &   4.54 &   4.54 &   4.54 &   4.54 \\
  3   &     -- &   3.19 &   3.19   &     -- &   2.99 &   2.99 &   3.00 &   3.00 \\ \hline
  \multicolumn{9}{c}{$p=4$, $h_4 = 12.73$} \\ \hline
  0   & 200.12 & 200.50 & 200.50   & 200.05 & 191.82 & 191.82 & 192.07 & 192.07 \\
  0.5 &  35.11 &  35.07 &  35.07   &  33.12 &  33.11 &  33.11 &  33.16 &  33.16 \\
  1   &  12.17 &  12.15 &  12.15   &  11.38 &  11.36 &  11.36 &  11.38 &  11.38 \\
  1.5 &   7.22 &   7.20 &   7.20   &   6.70 &   6.69 &   6.69 &   6.70 &   6.70 \\
  2   &   5.19 &   5.18 &   5.18   &   4.80 &   4.79 &   4.79 &   4.80 &   4.80 \\
  3   &   3.41 &   3.41 &   3.41   &   3.14 &   3.14 &   3.14 &   3.14 &   3.14 \\ \hline
  \multicolumn{9}{c}{$p=10$, $h_4 = 22.67$} \\ \hline
  0   & 199.95 & 200.77 & 200.76   & 200.06 & 190.38 & 190.38 & 190.72 & 190.72 \\
  0.5 &  48.52 &  48.54 &  48.54   &  44.19 &  45.17 &  45.17 &  45.27 &  45.27 \\
  1   &  15.98 &  15.93 &  15.93   &  14.32 &  14.47 &  14.47 &  14.51 &  14.51 \\
  1.5 &   9.23 &   9.21 &   9.21   &   8.23 &   8.21 &   8.21 &   8.24 &   8.24 \\
  2   &   6.57 &   6.56 &   6.56   &   5.83 &   5.77 &   5.77 &   5.79 &   5.79 \\
  3   &   4.28 &   4.28 &   4.28   &   3.79 &   3.70 &   3.70 &   3.71 &   3.71 \\ \hline
\end{tabular}
\end{table}
We added the \textit{zero-state} results to allow a comparison of both the potentially different accuracies between
\textit{zero-state} and \textit{steady-state} ARL and, of course, of the levels itselves for the considered $\delta$.
First, we conclude that \cite{Prab:Rung:1997} really determined the \textit{cyclical steady-state} ARL, $\mathcal{D}_\star$.
Second, we recognize similar accuracy differences between the Markov chain approach of \cite{Prab:Rung:1997} and
the methods deploying Nystr\"om with Gau\ss{}-Legendre quadrature for either ARL type.
The Monte Carlo confirmation runs with $10^9$ replications confirm the validity of the latter procedures.
Note that all Nystr\"om results are based on $r=30$ nodes resulting in linear equations systems of dimension 30 and 900, respectively.
Hence, the new method provides higher accuracy with smaller matrix dimensions which means less computing time.
Eventually, the differences between \textit{conditional} and \textit{cyclical steady-state} ARL
are little so that both could be applied for judging the long time behavior of
MEWMA control charts.

Next, we want to illustrate the dependence of the ARL to the shift magnitude $\delta$ utilizing all
three ARL types. Looking at the ARL maps we conjecture that the \textit{worst-case} ARL is realized for $\theta = \pi$ ($\gamma = -1$) and some
$\alpha$ close to the normalized threshold $h$. In the sequel we apply therefore golden section search to identify the final $\alpha$ yielding the
maximum $\mathcal{L}(\alpha, -1)$ from \eqref{eq:biLnew}.
As in Table~\ref{tab:02} we plot the ARL against $\sqrt{\delta} = \Vert \bm{\mu} \Vert \in (0,3.5)$.
\begin{figure}[hbtp]
\subfloat[$p=2$]{\includegraphics[width=.5\textwidth]{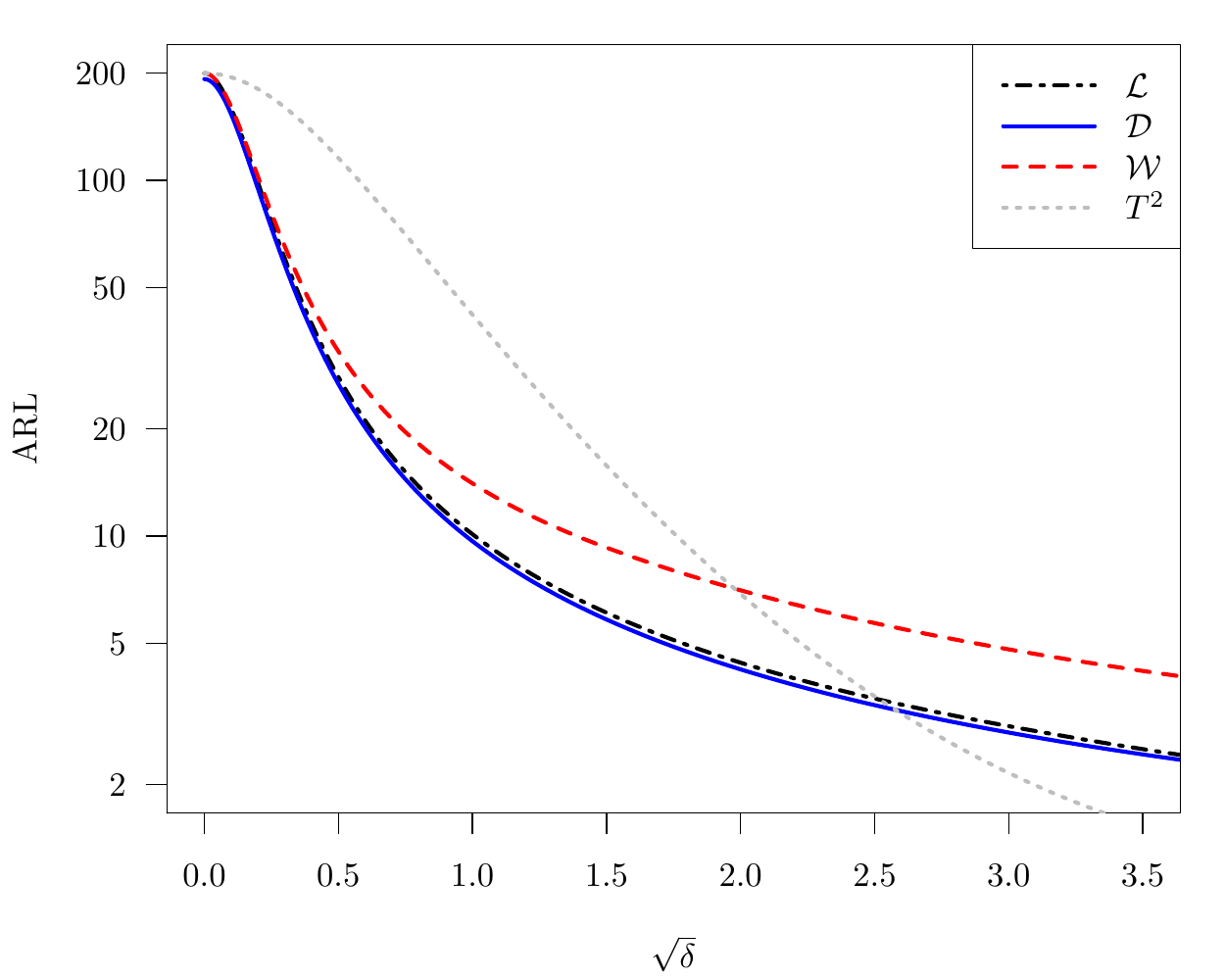}}
\subfloat[$p=3$]{\includegraphics[width=.5\textwidth]{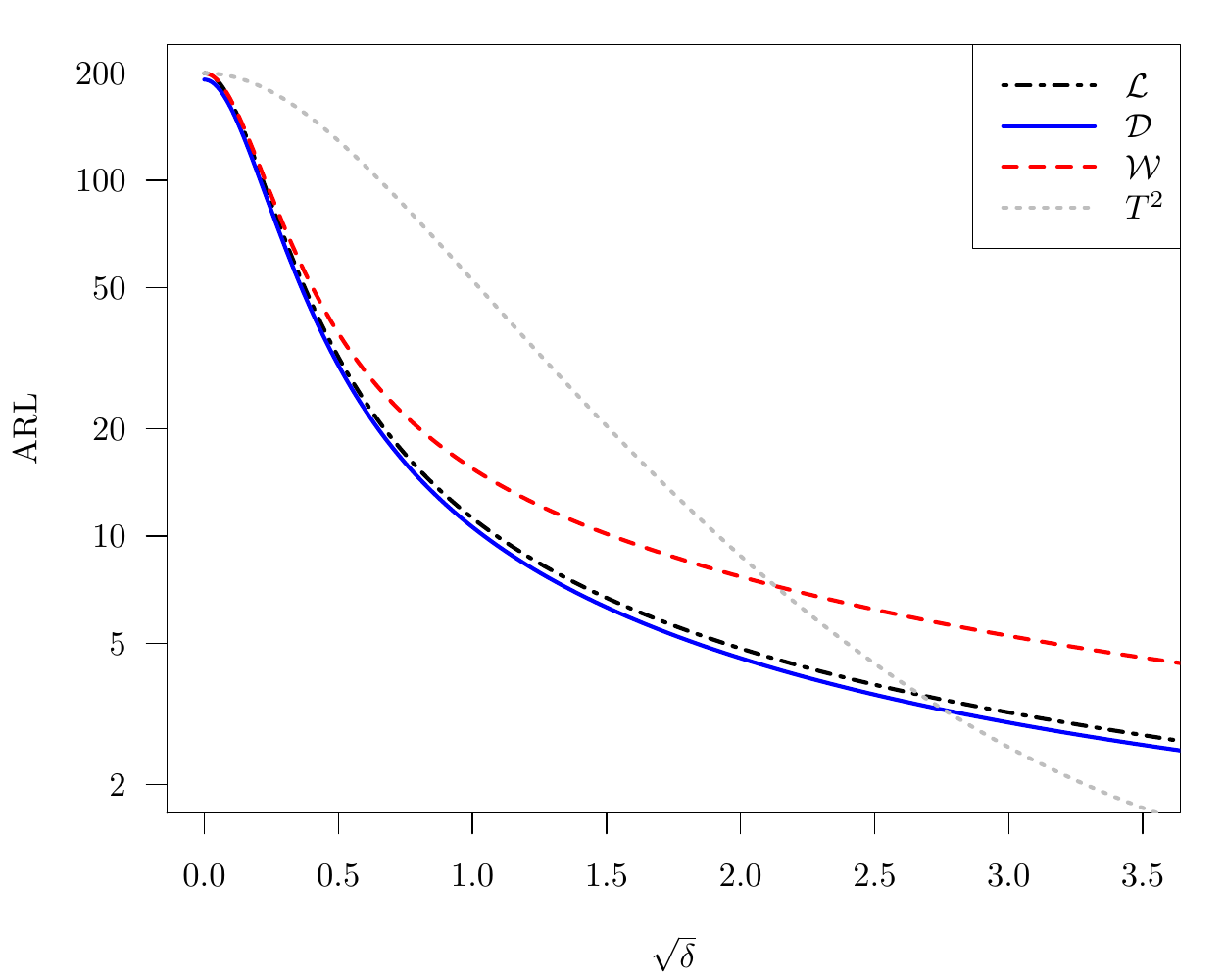}}
	
\subfloat[$p=4$]{\includegraphics[width=.5\textwidth]{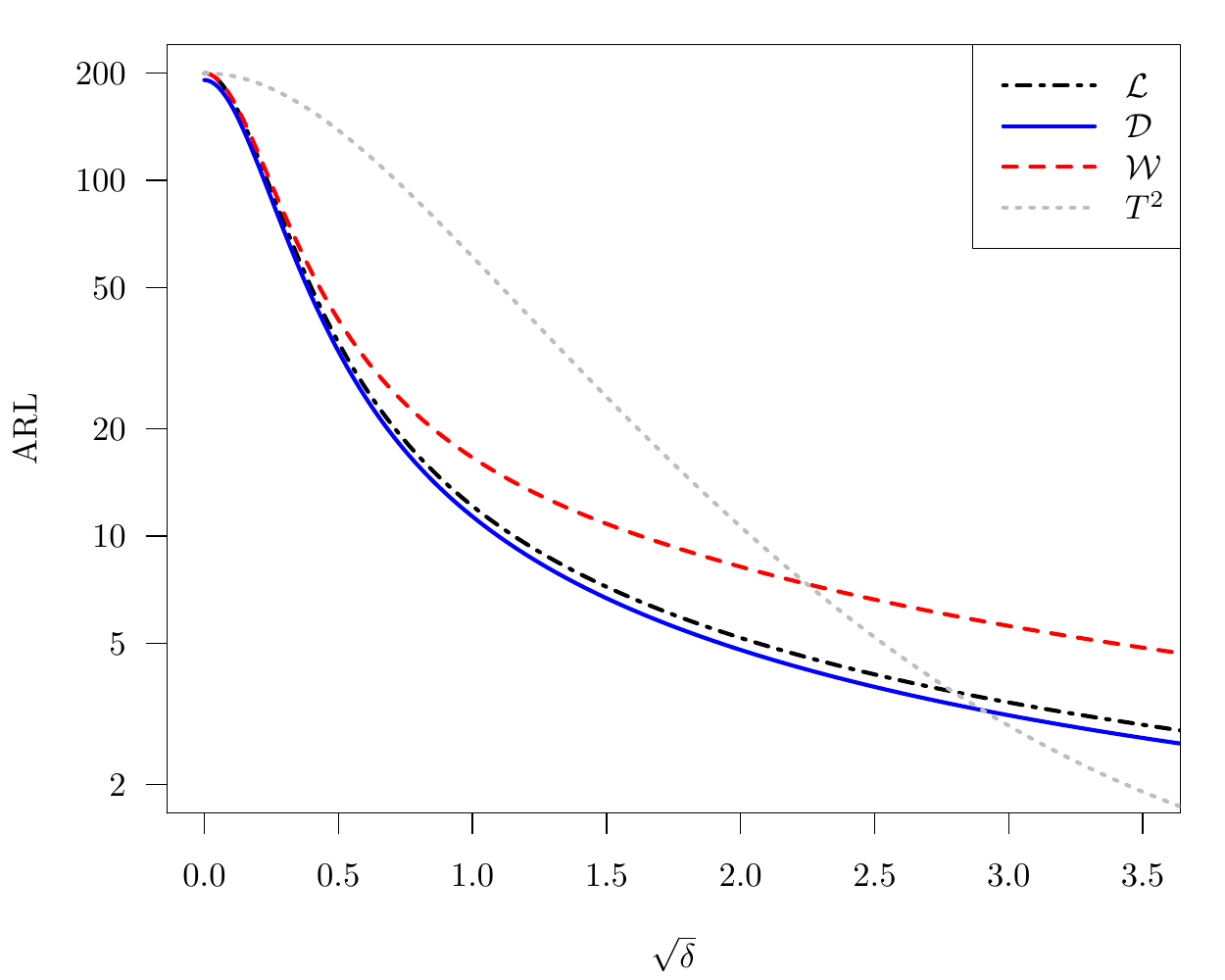}}
\subfloat[$p=10$]{\includegraphics[width=.5\textwidth]{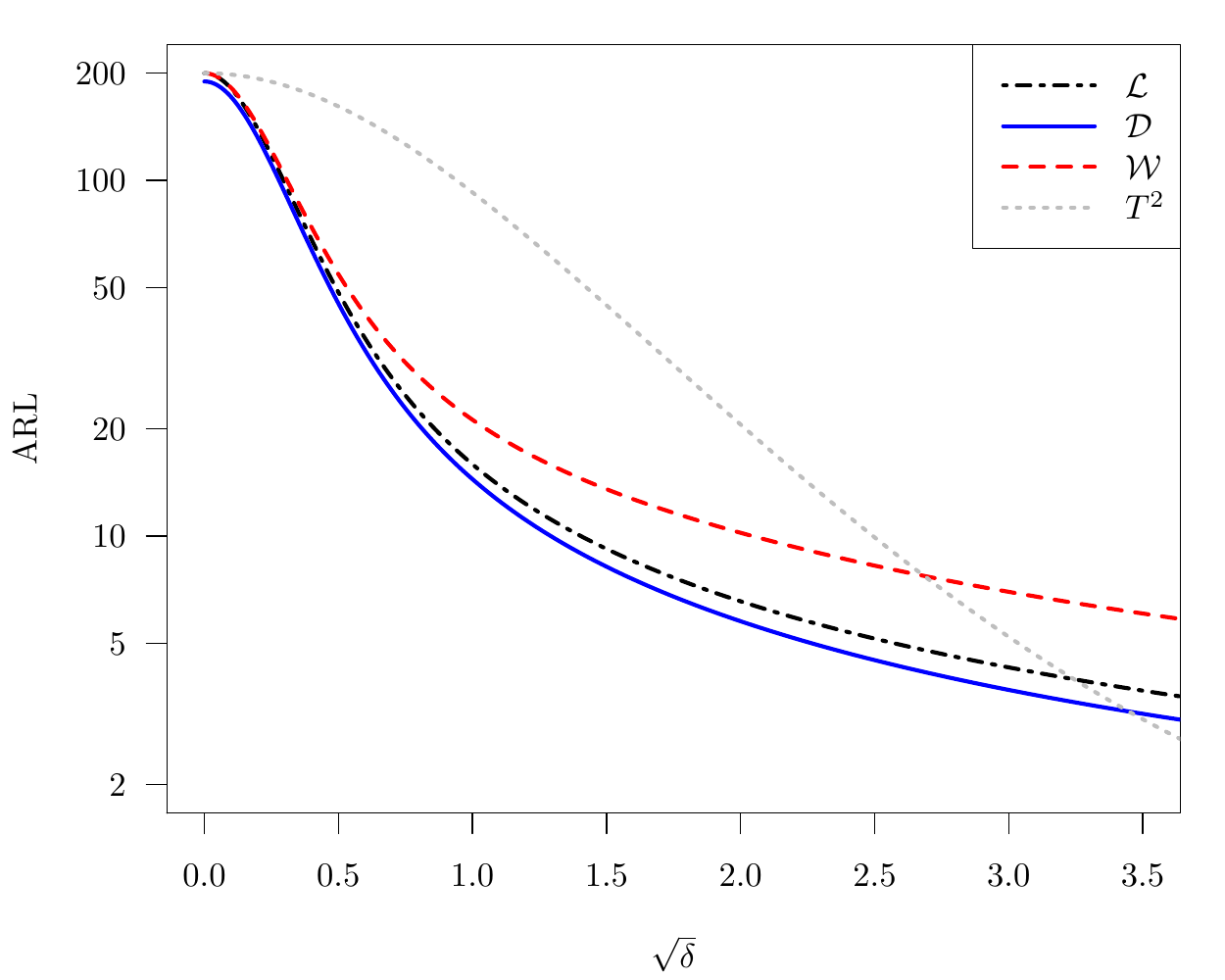}}
	
\caption{Diverse ARL (\textit{zero-state} $\mathcal{L}$, \textit{steady-state} $\mathcal{D}$, \textit{worst-case} $\mathcal{W}$) of MEWMA with $\lambda=0.1$,
and of Hotelling's ($T^2$) Shewhart-type chart vs. $\sqrt{\delta} = \sqrt{\Vert \bm{\mu} \Vert}$; $E_\infty(N) = 200$.}\label{fig:arlVSdelta1}
\end{figure}
In Figure~\ref{fig:arlVSdelta1}, we provide all three ARL types of MEWMA charts with $\lambda = 0.1$ and \textit{in-control} ARL 200.
In addition, we plot the respective single curves of the Shewhart-type Hotelling chart which corresponds to $\lambda = 1$ and
deploys only the most recent observation to decide whether signaling or not.

MEWMA clearly dominates the classic Hotelling chart for change magnitudes $\sqrt{\delta} < 2$.
In case of dimension $p=10$ it remains valid even for $\sqrt{\delta} < 2.5$.
Increasing the dimension beyond 10, we would see the MEWMA dominance over the whole interval $(0,3.5)$ which is not really surprising
because an increase in $p$ while holding $\delta$ means that the shift in relation to the vector length gets smaller.
Then control charts with memory such as MEWMA gain more and more in the competition with Shewhart charts.
Taking into account that the worst case is not very likely, we could
claim that MEWMA performs better for $\sqrt{\delta} < 2.5$ and $\sqrt{\delta} < 3$, respectively.

Comparing the \textit{zero-state} and the \textit{steady-state} ARL, we observe their divergence for increasing dimension $p$.
This is essentially driven by the subtle behavior of the \textit{steady-state} density of the MEWMA statistic, see Figure~\ref{fig:map2},
which counterbalances the increased difficulty of detecting a change of magnitude $\delta$ for increased $p$ by moving the probability mass
to favorable regions. The distance between the \textit{zero-state} and the \textit{worst-case} ARL
remains stable. Note that the ARL values are plotted on a log-scale, hence we observe constant ratios between $\mathcal{W}$ and $\mathcal{L}$.

In order to comprehend the influence of $\lambda$ to the ARL performance, we study the
relationship between $\lambda$ and the respective ARL type for one specific change, $\delta = 1$.
Ideally, we could derive some design rules as in \cite{Prab:Rung:1997}, Table 2, where the authors propose for $\delta = 1$ and $E_\infty(N)=500$
the values 0.105 and 0.085 for dimensions $p=4$ and $10$, respectively, aiming at minimal $\mathcal{L}$.
The slightly more precise numbers deploying the Gau\ss{}-Legendre Nystr\"om methods would be 0.104 and 0.086, respectively.
In Figure~\ref{fig:lambdaOpt} we illustrate the hunt for optimal $\lambda$ while minimizing all three ARL types separately.
Here we assume again $E_\infty(N)=200$ pointing out that the optimal $\lambda$ would be smaller for larger $E_\infty(N)$ as already indicated in \cite{Prab:Rung:1997}.
\begin{figure}[hbtp]
\subfloat[$p=2$]{\includegraphics[width=.5\textwidth]{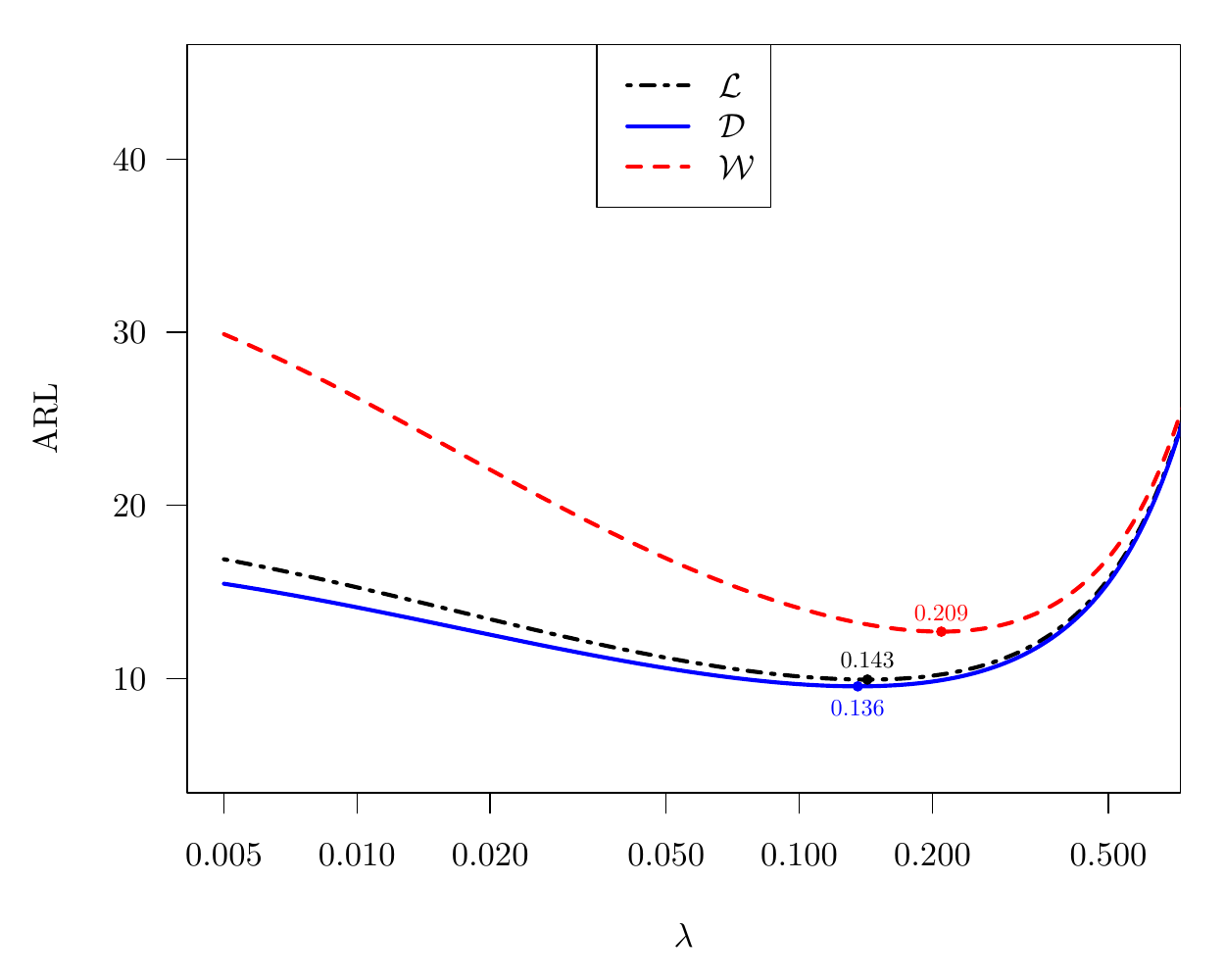}}
\subfloat[$p=3$]{\includegraphics[width=.5\textwidth]{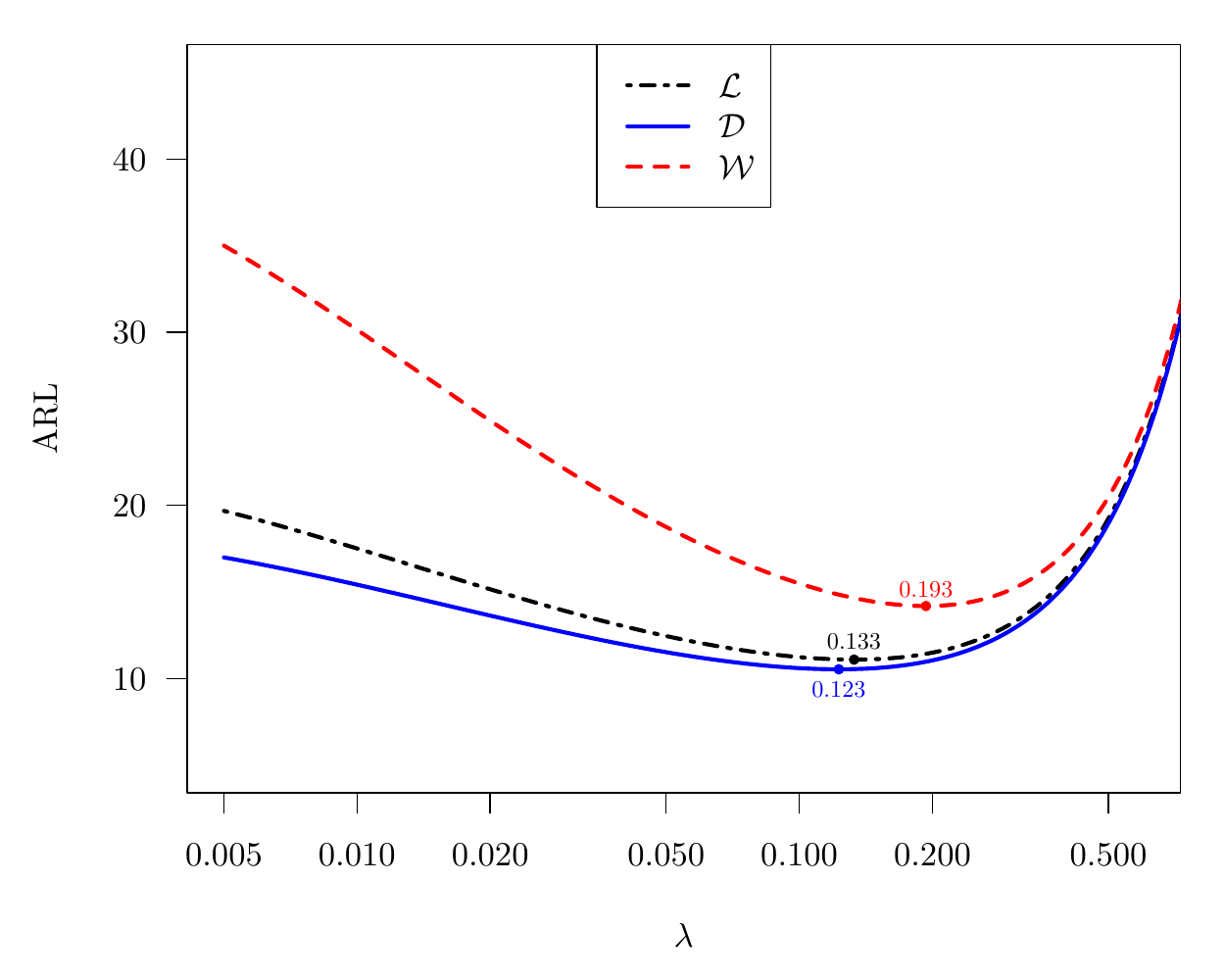}}
	
\subfloat[$p=4$]{\includegraphics[width=.5\textwidth]{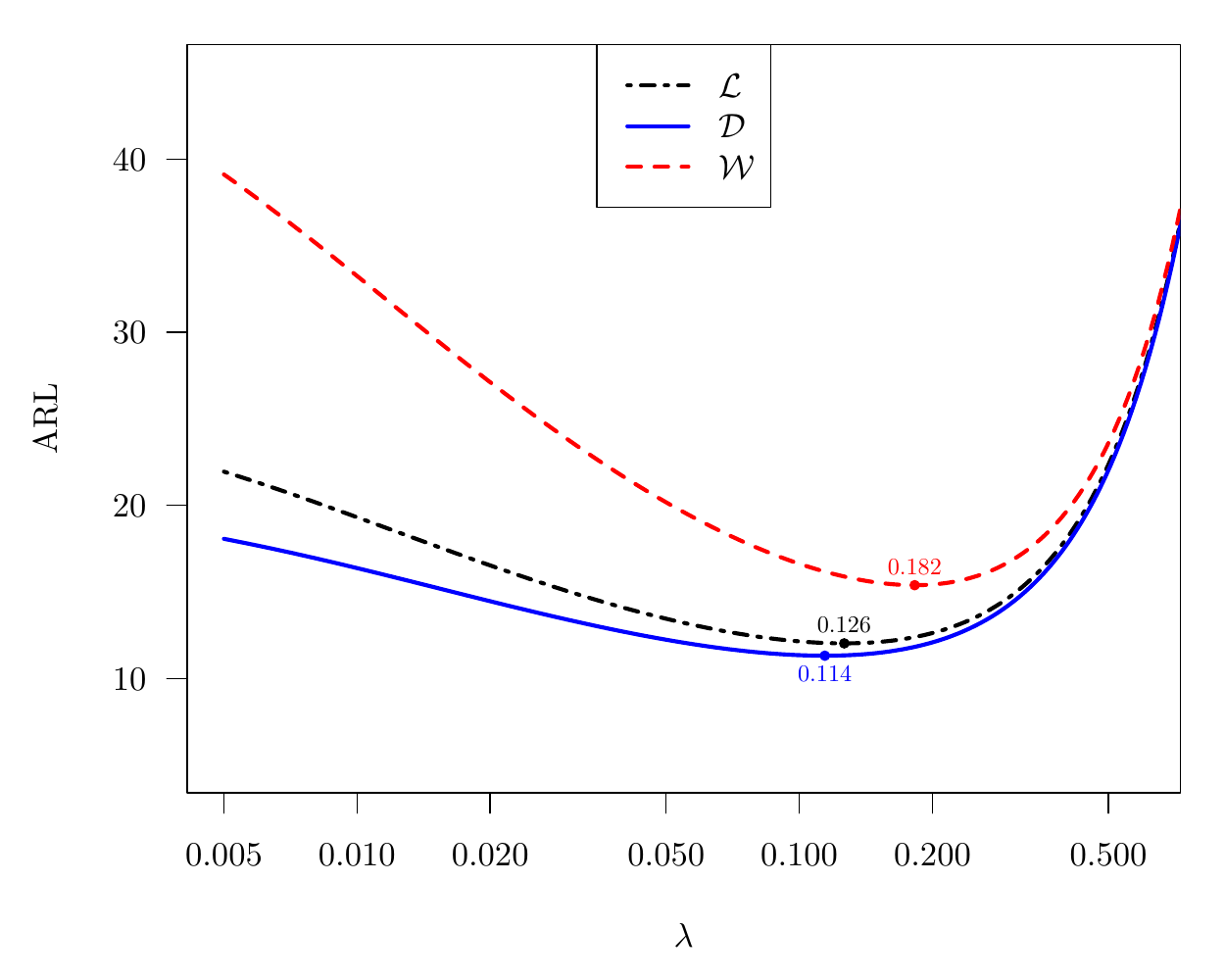}}
\subfloat[$p=10$]{\includegraphics[width=.5\textwidth]{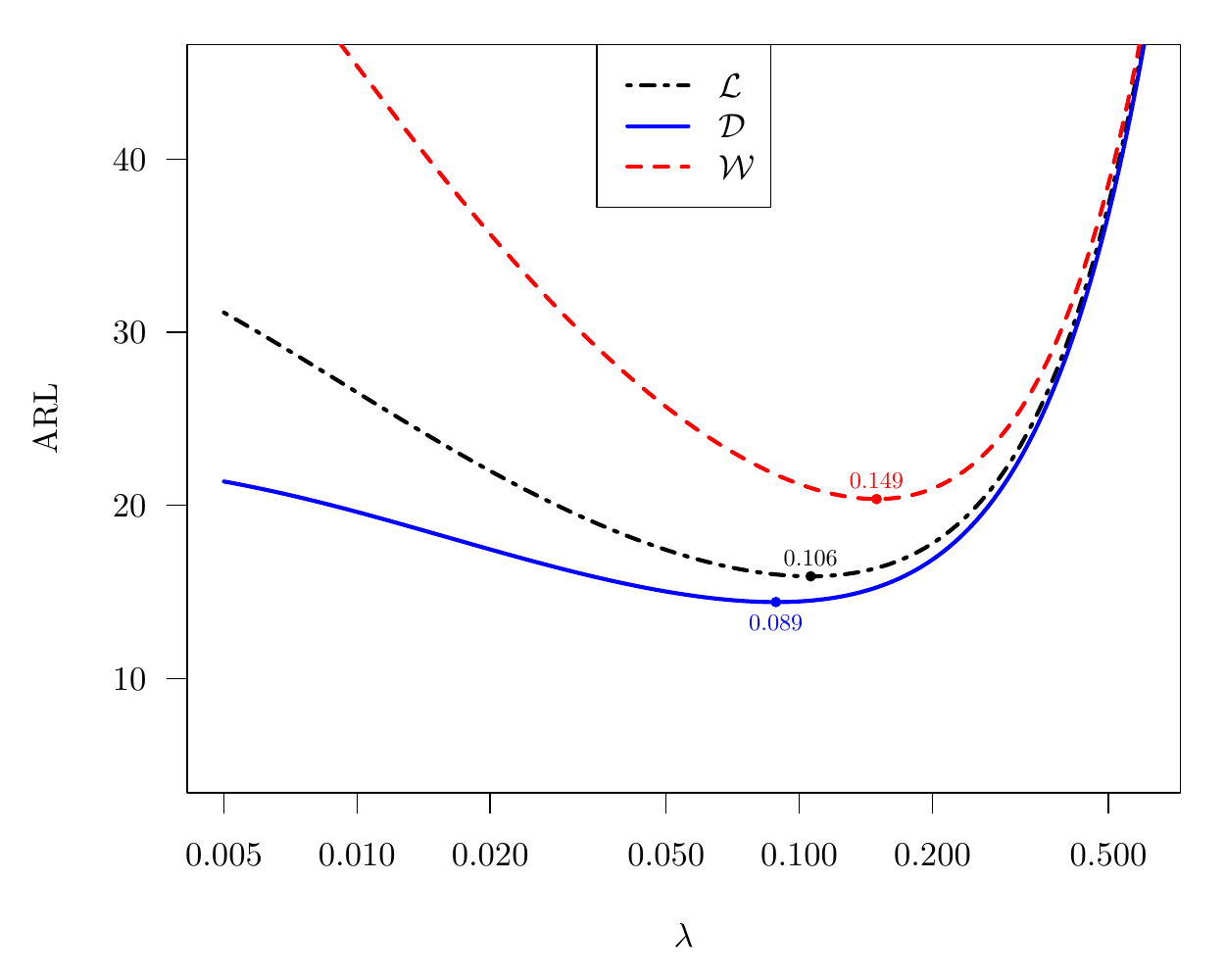}}
	
\caption{Striving for optimal $\lambda$ in terms of \textit{zero-state} $\mathcal{L}$, \textit{steady-state} $\mathcal{D}$, and \textit{worst-case} $\mathcal{W}$ of MEWMA
while detecting change from $\bm{0}$ to $\bm{\mu}_1$ with $\Vert \bm{\mu}_1 \Vert = 1$; $E_\infty(N) = 200$.}\label{fig:lambdaOpt}
\end{figure}
First, we recognize that the smallest $\lambda$ is obtained while minimizing the \textit{steady-state} ARL $\mathcal{D}$ closely followed by
the \textit{zero-state} ARL $\mathcal{L}$ one. The optimal $\lambda$ for the \textit{worst-case} ARL $\mathcal{W}$ is substantially larger.
We want to emphasize that all MEWMA curves are well below the corresponding ARL values of the Hotelling chart which are 41.9, 52.4, 61.0, and 92.5, respectively.
Even the (optimal) \textit{worst-case} results are substantially smaller (12.7, 14.2, 15.4, and 20.4, respectively).
Second, we observe like \cite{Prab:Rung:1997} that for increasing dimension $p$ the optimal $\lambda$ decreases for all three ARL types.
Interestingly, for large $p$ the profiles of $\mathcal{D}$ differ from the other ones considerably. It is even more pronounced, if $p$ becomes really large.
Going beyond $p = 30$, the profile is not convex for small $\lambda$ anymore. That is, the related values of $\mathcal{D}$ decrease with respect to $\lambda$ so
that the optimal $\lambda$ might be hidden behind $\lambda < 0.005$. This anomaly is known for dealing with minimizing $\mathcal{L}$ for one-sided
variance EWMA schemes -- see, e.\,g., \cite{Knot:2006a}. However, this time it is observed for the \textit{steady-state} ARL $\mathcal{D}$.
We remind that this peculiar behavior is caused by the patterns of the \textit{steady-state} density $\psi()$ for large $p$ illustrated in Figure~\ref{fig:map2}.
In summary, choices of $\lambda$ within $(0.1,0.2)$ which are popular in the univariate setup turn out to be also appropriate recommendations for
MEWMA. If $\delta$ is smaller or larger than 1, then, of course, $\lambda$ has to be decreased or increased accordingly.

\section{Conclusions}

In summary, the toolbox for calculating MEWMA ARL values is now complete. One could use either the neat Markov chain approximation introduced in \cite{Rung:Prab:1996}
and expanded in \cite{Prab:Rung:1997} for the dealing with the \textit{steady-state} ARL $\mathcal{D}$, or the highly specialized numerical algorithms for diverse integral equations proposed in
\cite{Rigd:1995a, Rigd:1995b}, modified in \cite{Knot:2017a}, and extended for $\mathcal{D}$ (and $\mathcal{D}_\star$) in this work.
For the second option, all needed routines are implemented in the R-package \texttt{spc}. We demonstrated the differences between the three
considered ARL types --- the classic \textit{zero-state} ARL $\mathcal{L}$, the \textit{conditional steady-state} ARL $\mathcal{D}$, and the
\textit{worst-case} ARL $\mathcal{W}$ --- and their different impact to the choice of the chart constant $\lambda$.
For large dimension $p$, eventually, we illustrated the odd behavior of the MEWMA statistic reaching the \textit{steady-state}.
Note that the decomposition idea in Lemma~\ref{LEM:01} could be utilized also to calculate the expected detection delays $D_\tau := E_\tau \big(N-\tau+1\mid N\ge \tau\big)$
for $\tau = (1,) 2, 3, \ldots$ The resulting sequence $\{D_\tau\}$ allows to evaluate the convergence patterns of $D_\tau \to \mathcal{D}$ and
consequently to judge the validity of the measure $\mathcal{D}$.

\appendix

\section{Linear equation systems}

First we plug in the Gau\ss{}-Legendre weights $w_i$ and nodes $z_i$ into \eqref{eq:L0igl} after replacing
$\alpha$ by $\alpha^2$, $u$ by $u^2$, and $\mathrm{d}u$ by $2u\,\mathrm{d}u$ so that we obtain
\begin{align*}
  \bm{\mathring{\ell}} & = (\mathring{\ell}_1, \ldots, \mathring{\ell}_r)^\prime
    \;,\;\; \mathring{\ell}_i = \mathcal{\mathring{L}}(z_i^2)
    \;,\;\; \mathbb{Q}_\mathcal{L} = (q_{ij}^\mathcal{L})_{i,j=1,\ldots,r} \;, \\
  & \qquad
     q_{ij}^\mathcal{L} =  w_j \frac{1}{\lambda^2} f_{\chi^2}\left(\frac{z_j^2}{\lambda^2} \,\Big|\, p, \eta z_i^2\right) 2 z_j \,, \\
  \bm{\mathring{\ell}} & = (\mathbb{I} - \mathbb{Q}_\mathcal{L})^{-1} \bm{1} \,,
  \intertext{and similarly for the left eigenfunctions}
  \bm{\mathring{\psi}} & = (\mathring{\psi}_1, \ldots, \mathring{\psi}_r)^\prime
    \;,\;\; \mathring{\psi}_i = \mathring{\psi}(z_i^2)
    \;,\;\; \mathbb{Q}_\psi = (q_{ij}^\psi)_{i,j=1,\ldots,r} \; \\
  & \qquad   
    q_{ij}^\psi =  w_j \frac{1}{\lambda^2} f_{\chi^2}\left(\frac{z_i^2}{\lambda^2} \,\Big|\, p, \eta z_j^2\right) 2 z_j \,, \\
  \varrho \bm{\mathring{\psi}} & = \mathbb{Q}_\psi \bm{\mathring{\psi}}  \qquad \text{... power method or standard eigenvalue procedure.} \\
  \bm{\mathring{\psi}}_\star & = (\mathring{\psi}_1^\star, \ldots, \mathring{\psi}_r^\star)^\prime
    \;,\;\; \mathring{\psi}_i^\star = \mathring{\psi}_\star(z_i^2)
    \;,\;\; \bm{\mathring{\psi}}_\star = (\mathbb{I} - \mathbb{Q}_\psi)^{-1} \bm{f} \qquad \text{... see \eqref{eq:psi0ble}.}  
\end{align*}
Note that $\bm{\mathring{\psi}}$ and $\bm{\mathring{\psi}}_\star$ rely on the same matrix $\mathbb{Q}_\psi$. Applying the Markov chain approximation
as in \cite{Prab:Rung:1997}, one would observe $\mathbb{Q}_\psi = \mathbb{Q}_\mathcal{L}^\prime$ for the corresponding Markov chain transition matrix.

\section{Transformation of ARL integral equation}

\begin{align*}
  \intertext{Change integration order in}
  \mathcal{L}(\alpha,\beta)
  & = 1 + \int_{-\sqrt{\delta h}}^{\sqrt{\delta h}} \int_{v^2/\delta}^h \mathcal{L}(u,v) K(u,v;\alpha,\beta) \,\mathrm{d}u\,\mathrm{d}v \\
  & = 1 + \int_0^h \int_{-\sqrt{\delta u}}^{\sqrt{\delta u}}  \mathcal{L}(u,v) K(u,v;\alpha,\beta) \,\mathrm{d}v\,\mathrm{d}u \,, \\
  K(u,v;\alpha,\beta) & = \frac{1}{\sqrt{2\pi\delta\lambda^2}}
  e^{-\frac{[v-\lambda\delta-(1-\lambda)\beta]^2}{2\delta\lambda^2}} 
  \frac{1}{\lambda^2} f_{\chi^2}\left(\!\frac{u-v^2/\delta}{\lambda^2} \,\Big|\, p-1, \eta (\alpha-\beta^2/\delta)\!\right) \,. 
  \intertext{Change second argument}
  \beta & = \sqrt{\alpha \delta} \, \gamma \;,\;\; v = \sqrt{u \delta} \, w \;,\;\;
  \mathrm{d}v = \sqrt{u \delta} \, \mathrm{d}w \;,\;\; w = v / \sqrt{\delta u} \,, \\ 
  \intertext{so that}
  \mathcal{L}(\alpha,\gamma) & = 1 + \int_0^h \int_{-1}^1 \mathcal{L}(u,w) K^\dag(u,w;\alpha,\gamma) \,\mathrm{d}w\,\mathrm{d}u \,, \\
  K^\dag(u,w;\alpha,\gamma) & = \frac{\sqrt{u}}{\sqrt{2\pi\lambda^2}}
  e^{-\frac{[\sqrt{u} w - \lambda\sqrt{\delta} - (1-\lambda)\sqrt{\alpha} \gamma]^2}{2\lambda^2}} \\
  & \qquad \times
  \frac{1}{\lambda^2} f_{\chi^2}\left(\frac{u(1-w^2)}{\lambda^2} \,\Big|\, p-1, \eta \alpha(1-\gamma^2) \right) \,.
\end{align*}

\section{Numerics of ARL integral equation}

Let $(z^{(0)}_i, w^{(0)}_i)$ and $(z^{(1)}_j, w^{(1)}_j)$ be the quadrature nodes and weights on $[0,h]$ and $[-1,1]$, respectively.
Then we solve the following linear equation system(s).
\begin{align*}
 \mathcal{L}_{ij} & = \mathcal{L}(z^{(0)}_i, z^{(1)}_j) \\
  & = 1 + \frac{1}{\lambda^3\sqrt{2\pi}} \sum_{k=1}^N w^{(0)}_k \sqrt{z^{(0)}_k} \sum_{l=1}^N  w^{(1)}_l
  e^{-\frac{\left[\sqrt{z^{(0)}_k} z^{(1)}_l - \lambda\sqrt{\delta} - (1-\lambda)\sqrt{z^{(0)}_i} z^{(1)}_j\right]^2}{2\lambda^2}} \ldots \\
  & \qquad \times f_{\chi^2}\left(\frac{z^{(0)}_k\big(1-(z^{(1)}_l)^2\big)}{\lambda^2} \,\Big|\, p-1, \eta z^{(0)}_i\big(1-(z^{(1)}_j)^2\big) \right) \,. \\
%
  %
  \intertext{w/ $u$ to $u^2$ and $w$ to $\sin(w)$ in \eqref{eq:biLnew}:}
  &  = 1 + \frac{1}{\lambda^3\sqrt{2\pi}} \sum_{k=1}^N w^{(0)}_k 2 (z^{(0)}_k)^2\sum_{l=1}^N  w^{(1)}_l \cos(z^{(1)}_l) \,
  e^{-\frac{\left[z^{(0)}_k \sin(z^{(1)}_l) - \lambda\sqrt{\delta} - (1-\lambda)z^{(0)}_i \sin(z^{(1)}_j)\right]^2}{2\lambda^2}} \ldots \\
  & \qquad \times f_{\chi^2}\left(\frac{(z^{(0)}_k)^2\big(\cos(z^{(1)}_l)\big)^2}{\lambda^2} \,\Big|\, p-1, \eta (z^{(0)}_i)^2\big(\cos(z^{(1)}_j)\big)^2 \right) \,.
\end{align*}

\section{Angle distribution}

From standard math literature we obtain for the surface area on the unit sphere $S^{p-1}$
\begin{equation*}
 A_{p-1} = \frac{2 \pi^{\frac{p}{2}}}{\Gamma\left(\frac{p}{2}\right)} \;,\;\; p = 1, 2, \ldots
\end{equation*}
This surface is assembled by a continuous set of circles of latitude which are spheres of one dimension less whose
radius depends on the latitude. Their area depending on latitude $\theta \in [0,\pi]$ with related radius $\sin(\theta)$ follows
\begin{equation*}
  \tilde{A}_{p-2}(\theta) = \frac{2 \pi^{\frac{p-1}{2}}}{\Gamma\left(\frac{p-1}{2}\right)} \sin(\theta)^{p-2} \;,\;\; p = 2, 3, \ldots
\end{equation*}
Now we derive the density of $\theta$ using the proportion of $\tilde{A}_{p-2}(\theta)$ relative to $A_{p-1}$:
\begin{equation*}
   d(\theta)
   = \frac{\tilde{A}_{p-2}(\theta)}{A_{p-1}}
   = \frac{\Gamma\left(\frac{p}{2}\right)}{\Gamma\left(\frac{p-1}{2}\right)\sqrt{\pi}} \sin(\theta)^{p-2} \,.
\end{equation*}
Applying the transformation $\gamma = - \cos(\theta)$ we get
\begin{align*}
  \theta & = \text{acos}(-\gamma)
    \;,\;\; \frac{\mathrm{d}}{\mathrm{d}\theta} \gamma(\theta) = \sin(\theta) = \sin\big(\text{acos}(-\gamma)\big) = \sqrt{1-\gamma^2} \,. \\
  d(\gamma)
  & = \frac{\Gamma(\frac{p}{2})}{\Gamma(\frac{p-1}{2})} (1-\gamma^2)^{\frac{p-3}{2}} / \sqrt{\pi}
    \;,\;\; \gamma \in (-1,1] \,.
\end{align*}

\section{Proof supporting Lemma~\ref{LEM:01}}

First, we make use of the following presentation of the non-central $\chi^2$ density, which was mentioned already in \cite{Vena:1973a}, equation (2.10):
\begin{equation*}
  f_{\chi^2}(x\mid p, \nu) = e^{-(x+\nu)/2} \frac{x^{p/2-1}}{2^{p/2} \Gamma(p/2)} \, {_0} F_1(; p/2; \nu x/4) \,.
\end{equation*}
Thereby, ${_0} F_1(;b;z)$ is called confluent hypergeometric limit function and is closely related to Bessel functions.
To get an idea about ${_0} F_1()$, we give one presentation:
\begin{equation*}
  {_0} F_1(;b;z) = \sum_{n=0}^\infty \frac{z^n}{(b)_n\, n!} \; \text{ with } \; (b)_n = b (b+1) \cdots (b+n-1) \;\; \text{\small (Pochhammer symbol)} \,.
\end{equation*}
A more rigor discussion with proofs is given in \cite{Muir:1982}, Theorem 1.3.4.
Taking these subtleties aside, we start with
\begin{align*}
  \lefteqn{\lambda^2 (1-\gamma^2)^{\frac{p-3}{2}} K^\dag(u,w;\alpha,\gamma)} \\
  & = (1-\gamma^2)^{\frac{p-3}{2}} \frac{\sqrt{u}}{\sqrt{2\pi\lambda^2}}
  e^{-\frac{[\sqrt{u} w - (1-\lambda)\sqrt{\alpha} \gamma]^2}{2\lambda^2}} 
  f_{\chi^2}\left(\frac{u (1 - w^2)}{\lambda^2} \,\Big|\, p-1, \eta \alpha(1-\gamma^2) \right) \\
  & = (1-\gamma^2)^{\frac{p-3}{2}} \frac{\sqrt{u}}{\sqrt{2\pi\lambda^2}}
  e^{-\frac{u w^2 -2 \sqrt{u} w (1-\lambda) \sqrt{\alpha} \gamma + (1-\lambda)^2\alpha\gamma^2}{2\lambda^2}} \\
  & \qquad \times
  e^{-\frac{u(1-w^2) +(1-\lambda)^2 \alpha (1-\gamma^2)}{2\lambda^2}}
  \frac{ \left( \frac{u(1-w^2)}{\lambda^2} \right)^{p/2-1} }{2^{p/2} \Gamma(p/2)} \\
  & \qquad\qquad \times 
  {_0 F_1}\left(; \frac{p-1}{2}; \frac{u (1 - w^2) (1-\lambda)^2 \alpha(1-\gamma^2)}{4\lambda^4}\right) \,.
\end{align*}
Rearrange variables:
\begin{align*}
%
  & = (1-w^2)^{\frac{p-3}{2}} \frac{\sqrt{u}}{\sqrt{2\pi\lambda^2}}
  e^{-\frac{u \gamma^2 -2 \sqrt{u} \gamma (1-\lambda) \sqrt{\alpha} w + (1-\lambda)^2\alpha w^2}{2\lambda^2}} \\
  & \qquad \times
  e^{-\frac{u(1-\gamma^2) +(1-\lambda)^2 \alpha (1-w^2)}{2\lambda^2}}
  \frac{ \left( \frac{u(1-\gamma^2)}{\lambda^2} \right)^{p/2-1} }{2^{p/2} \Gamma(p/2)} \\
  & \qquad\qquad \times 
  {_0 F_1}\left(; \frac{p-1}{2}; \frac{u (1 - \gamma^2) (1-\lambda)^2 \alpha (1-w^2)}{4\lambda^4}\right) \\
  & = (1-w^2)^{\frac{p-3}{2}} \frac{\sqrt{u}}{\sqrt{2\pi\lambda^2}}
  e^{-\frac{[\sqrt{u} \gamma - (1-\lambda)\sqrt{\alpha} w]^2}{2\lambda^2}}
  f_{\chi^2}\left(\frac{u (1 - \gamma^2)}{\lambda^2} \,\Big|\, p-1, \eta \alpha(1-w^2) \right) \,.
\end{align*}
Now we are ready to perform the integration:
\begin{align*}
  \lefteqn{\left((1-w^2)^{\frac{p-3}{2}}\right)^{-1} \lambda^2 \int_{-1}^1 (1-\gamma^2)^{\frac{p-3}{2}} K^\dag(u,w;\alpha,\gamma) \,\mathrm{d}\gamma} \\
  & = \int_{-1}^1 \frac{\sqrt{u}}{\sqrt{2\pi\lambda^2}} e^{-\frac{[\sqrt{u} \gamma - (1-\lambda)\sqrt{\alpha} w]^2}{2\lambda^2}}
    f_{\chi^2}\left(\frac{u (1 - \gamma^2)}{\lambda^2} \,\Big|\, p-1, \eta \alpha(1-w^2) \right) \,\mathrm{d}\gamma \\
  & = \int_{-\sqrt{u}}^{\sqrt{u}} \frac{1}{\sqrt{2\pi\lambda^2}} e^{ -\frac{ [\tilde{\gamma} - (1-\lambda)\sqrt{\alpha} w]^2}{2\lambda^2} }
    f_{\chi^2}\left(\frac{u - \tilde{\gamma}^2}{\lambda^2} \,\Big|\, p-1, \eta \alpha(1-w^2) \right) \, \mathrm{d}\tilde{\gamma} \\
  & = f_{\chi^2}\left(\frac{u}{\lambda^2} \,\Big|\, p, \eta \alpha \right) \,.
\end{align*}
The last integral follows from summing two $\chi^2$ variates, $G \sim \mathcal{N} \big((1-\lambda)\sqrt{\alpha} w, \lambda^2 \big)$
and $V/\lambda^2 \sim \chi^2_{p-1, \eta \alpha (1-w^2)}$. Then for the sum we observe
$(G^2 + V)/\lambda^2 \sim \chi^2_{p, \eta \alpha}$.
Finally, moving the terms before the integral to the right-hand side yields
\begin{align*}
  \int_{-1}^{1} (1-\gamma^2)^{\frac{p-3}{2}} K^\dag(u,w;\alpha,\gamma) \,\mathrm{d}\gamma
  & = (1-w^2)^{\frac{p-3}{2}} \frac{1}{\lambda^2} f_{\chi^2}\left(\frac{u}{\lambda^2} \,\Big|\, p, \eta \alpha\right) \,.
\end{align*}

\bibliographystyle{apalike}

\bibliography{/home/knoth/common/references/sk}



\vspace*{2ex}

\noindent\makebox[\linewidth]{\rule{\textwidth}{0.4pt}}

\vspace*{2ex}

\end{document}